\documentclass[12pt]{amsart}
\usepackage{latexsym,amssymb}
\usepackage[foot]{amsaddr}
\usepackage{graphics}
\usepackage{url}
\usepackage{harvard}
\usepackage{eulervm}
\usepackage{epsfig,subfigure}
\usepackage[draft,silent]{fixme}
\usepackage{tikz}
\usepackage{graphicx}
\usepackage[ruled]{algorithm2e}
\usepackage{tikz}
\usepackage{tkz-euclide,tikz}
\usepackage{setspace}
\usepackage[margin=1.1in]{geometry} %for double spaced (jasa) versi

\renewcommand{\P}{\mathbb{P}}
\newtheorem{example}{Example}
\newtheorem{remark}{Remark}

\def\argmax{\text{argmax}}

\def\HH{\mathcal{H}}
\def\KK{\mathcal{K}}
\def\RR{\mathbb{R}}
\def\PP{\mathbb{P}}
\def\OO{\mathcal{O}}
\def\SS{\mathcal{S}}
\def\ZZ{\mathcal{Z}}
\newtheorem{lemma}{Lemma}
\newtheorem{theorem}{Theorem}

\newtheorem{assumption}{Assumption}
\renewenvironment{proof}{\noindent {\bf Proof.\ }}{\hfill{\rule{2mm}{2mm}}}
\renewenvironment{remark}{\noindent {\bf Remark.\ }} {\hfill{\rule{2mm}{2mm}}}
\title[Random Coefficients for Binary Response]{Nonparametric maximum likelihood methods\\ 
for binary response models with random coefficients}
\author{Jiaying Gu}
\address{Department of Economics, University of Toronto, Toronto, Ontario, M5S 3G7, Canada}
\author{Roger Koenker}
\address{Department of Economics, University College London, London, WC1H OAX, UK}
\date{\today} 

\begin{document}
\bibliographystyle{econometrica}

\begin{abstract}
    The venerable method of maximum likelihood has found numerous recent applications in
    {\it nonparametric} estimation of regression and shape constrained densities. For mixture models
    the nonparametric maximum likelihood estimator (NPMLE) of \citeasnoun{KW} plays a central 
    role in recent developments of empirical Bayes
    methods.  The NPMLE has also been proposed by \citeasnoun{Cosslett83} 
    as an estimation method for single index linear models for binary response with random coefficients.
    However, computational difficulties have hindered its application.  Combining recent developments in
    computational geometry and convex optimization we develop a new approach to computation for such
    models that dramatically increases their computational tractability.  Consistency of the method
    is established for an expanded profile likelihood formulation.  The methods are
    evaluated in simulation experiments, compared to the deconvolution methods of \citeasnoun{GK}  
    and illustrated in an application to modal choice for journey-to-work data in the Washington DC area.

\end{abstract}
\maketitle
%\doublespacing

\section{Introduction}

Statistical models with random coefficients go back at least to \citeasnoun{Fisher}.  Subsequent work
following \citeasnoun{Henderson} and \citeasnoun{Scheffe} has focused primarily on parametric models:
in these hierarchical, or mixed models, random coefficients appear with explicit parametric forms 
specified for their distributions and estimation focused on the hyperparameters of these distributions.
Early on this took the form of linear models with Gaussian random coefficients for which there were
efficient computational methods based on data augmentation.
More recent developments, stimulated by the Bayesian literature and the advent of Markov chain Monte-Carlo
methods have considerably expanded the menu of distributional choices for such models, but the focus
on parametric specifications has remained.  In the binary response literature there has also been
some interest in delving beyond the usual generalized linear model formalism with fixed link functions
like logit or probit, but again random coefficients are typically assumed to take particular 
parametric forms, often Gaussian. 
% we might need some references here, but I'm not sure what; there is Chesher & Rosen, or Stan, or \ldots
Our objective in this work has been to explore nonparametric random coefficient models for binary
response thereby relaxing the rather stringent parametric assumptions in most of the prior literature.
As in other nonparametric settings our objective can be motivated by a healthy scepticism 
about a priori parametric specifications and the general spirit of exploratory data analysis. 
Our approach, extending the work of \citeasnoun{Cosslett83}, exploits an exact likelihood formulation
that shares some common features with recent developments in nonparametric maximum likelihood estimation
of mixture models.  There is also a close connection to related work on random coefficient models for
continuous univariate response in the medical imaging literature, \citeasnoun{BFH} and \citeasnoun{FV00}.

We will consider the linear index, random coefficient binary response model, 
\begin{equation}
y_i = 1\{ x_i ^\top \beta_i+ w_i^\top \theta_0 \geq 0\}.  \label{eq: model}
\end{equation}
We observe covariates $x_i \in \mathbb{R}^{d+1}$, $w_i \in \mathbb{R}^p$, and the binary response, $y_i$.
The indicator function, $1\{E\}$, takes the value 1 when $E$ is true, and 0 otherwise. 
We will assume that  the random parameters $\beta_i$ are independent 
of both $x_i$ and $w_i$, and that the $\beta_i$ are independent and identically distributed with an unknown distribution $F_0$.  
The remaining Euclidean parameters, $\theta_0 \in \Theta \subset \RR^p$ are fixed and also unknown.   
Our objective is to estimate the pair $(\theta_0,  F_0)$ by maximum likelihood.
This model encompasses many existing single index binary choice models in the literature. 
When the covariate vector $x_i$ contains only an intercept term so $\beta_i$ is a random scalar 
with known distribution, either logistic or Gaussian, we have the logit or probit model respectively.
When $\beta_i$ is scalar with unknown distribution we have the simplest version of the semiparametric 
single index model of \citeasnoun{KleinSpady}  and \citeasnoun{Cosslett83}.  
When $d \geq 1$ and there are no other covariates 
$w_i$, we have the random coefficient single index model of \citeasnoun{Ichimura},
and \citeasnoun{GK}.

It is immediately apparent that the distribution of the $\beta_i$'s
is only identified up to a scale transformation of the coordinates of $\beta$, so without loss of 
generality we could impose the normalization that
$\|\beta \| = 1$.  An additional identification requirement noted by \citeasnoun{Ichimura} is
that the distribution of $\beta_i$ must have support on some hemisphere. This requirement will be fulfilled 
if the sign of one of the $\beta$ coordinates is known, we will assume henceforth that the last 
entry of $\beta$ is positive, so $\beta_i$'s would be restricted to lie in the northern hemisphere.
Under this additional assumption an alternative normalization simply takes the coordinate with the known 
sign to be 1 and focuses attention on the joint distribution of the remaining coordinates relative to it. 
In our modal choice application for example the price effect can be normalized to 1, and
remaining covariate effects are interpreted relative to the effect of price.
Henceforth, we shall assume that $x_i = (1, z_i^\top, -v_i)$ and $\beta_i = (\eta_i^\top , 1)$ with $\eta \in \RR^d$,
and estimation $F_0 = F_\beta$ is reduced to estimation of $F_0 = F_\eta$.
Finally, it should be stressed from the outset that identification requires sufficient variability 
in $x_i$ to trace out $F_0$ on its full support.  These conditions will be made more explicit in 
the sequel.

We begin with a brief discussion of the simplest case in which there is only a univariate random 
coefficient as considered in the seminal paper of \citeasnoun{Cosslett83},  a formulation that already 
illustrates many of the essential ideas.   This is followed by a general treatment of the multivariate 
setting that draws on recent developments in combinatorial geometry involving ``hyperplane arrangements.''  
A discussion of
identification and consistency is then followed by a brief description of some simulation experiments.
Performance comparisons are made with the deconvolution approach of \citeasnoun{GK}.
We conclude with an illustrative application to modal choice of commuters in the Washington DC area 
based on data from \citeasnoun{Horowitz93}.  Proofs of all formal results are collected in Appendix A.

\section{Univariate Randomness}

In our simplest setting we have only a univariate random component and we observe thresholds, $v_i$,
and associated binary responses,
\begin{equation}
    \label{eq.Coss1}
    y_i = 1 (\eta_i \geq v_i) \quad i = 1, \ldots , n,
\end{equation}
with $\eta_i$'s drawn iidly from the distribution $F_\eta$ and independently of the $v_i$.  
%In terms of the prior notation we have $\beta_i = -\eta_i $,  $\tilde X_i = 0$, $W_i = v_i$ and $\theta_01$. 
In economic applications, with $v_i$ taken as a price the survival curve, $1 - F_\eta (v) $ can be interpreted
as a demand curve, the proportion of the population willing to pay at pricce, $v$.  More generally, the single index model, 
might express $v_i = v(w_i , \theta)$ depending on other covariates, $w_i$ and unknown parameters, $\theta$.  
This is the context of \citeasnoun{Cosslett83} who focuses 
most of his attention on estimation of the distribution $F_\eta$ employing the nonparametric
maximum likelihood estimator (NPMLE) of \citeasnoun{KW}.   Estimation of the remaining parameters can be carried
out by some form of profile likelihood, but we will defer such considerations. 
%to Section \ref{sec.profile}.
In biostatistics \eqref{eq.Coss1} is referred to as the current status model:  with the inequality 
reversed we observe inspection times,
$v_i$, and a binary indicator, $y_i$, revealing whether an unobserved event time, $\eta_i$, has occurred 
prior to its associated inspection time, $v_i$.  Again, the objective is to estimate the distribution of 
the event times, $F_\eta$, by nonparametric maximum likelihood as described, for example, by \citeasnoun{GJW}.  

The geometry of maximum likelihood in this univariate setting is quite simple and helps to establish a
heuristic for the general multivariate case.  To illustrate the critical role of the intervals between
adjacent order statistics of the $v_i$ under the standard convention of the current status model, we can write,
\[
 \PP (y = 1 | v) = \int 1 (\eta \leq v) dF_\eta (\eta).
\]
Given a sample $\{ (y_i, v_i) \; i = 1, \ldots , n\}$, this relation defines $n+1$ intervals as illustrated
for $n = 10$ in the upper panel of Figure \ref{fig.Oned}.  
Each point defines a half line on $\RR$; if $y_i = 1$ then the
interval is $R_i = (-\infty, v_i ]$, while if $y_i = 0$ the interval is $R_i = (v_i , \infty )$.  
Neglecting for the moment the possibility of ties, 
the $R_i$'s form a partition of $n + 1$ intervals of $\RR$, 
We will denote these intervals by $I_j$, for $j = 1, \ldots , n+1$; they can be either closed, open or
half open.  We now define the associated counts,
\begin{align*}
c_j & = \sum_{i=1}^n 1\{R_i \cap I_j \neq \emptyset\}\\
    & =\begin{cases}
\sum_{i: y_i = 0} 1\{v_i < v_{(j)}\} + \sum_{i: y_i = 1} 1\{v_i \geq v_{(j)}\} & \mbox{for}\quad  1\leq j \leq n\\
\sum_{i: y_i = 0} 1\{v_i \leq v_{(n)}\} & \mbox{for} \quad  j = n+1
\end{cases}
\end{align*}
where $v_{(j)}$ is the $j$-th order statistic of the sample $\{v_1, \dots, v_n\}$. 
Now suppose we assign probability mass, $p_j$ to each of the intervals $I_j$, so that 
$p = \{p_1,\dots, p_{n+1}\}$ in the $n$ dimensional unit simplex, $\SS_n$.  
Then the contribution of the $i$-th data point to the likelihood is 
$\sum_{j} p_j 1\{R_i \cap I_j\neq \emptyset\}$.
\begin{figure}[bp]
    \begin{center}
    \resizebox{.65\textwidth}{!}{\includegraphics{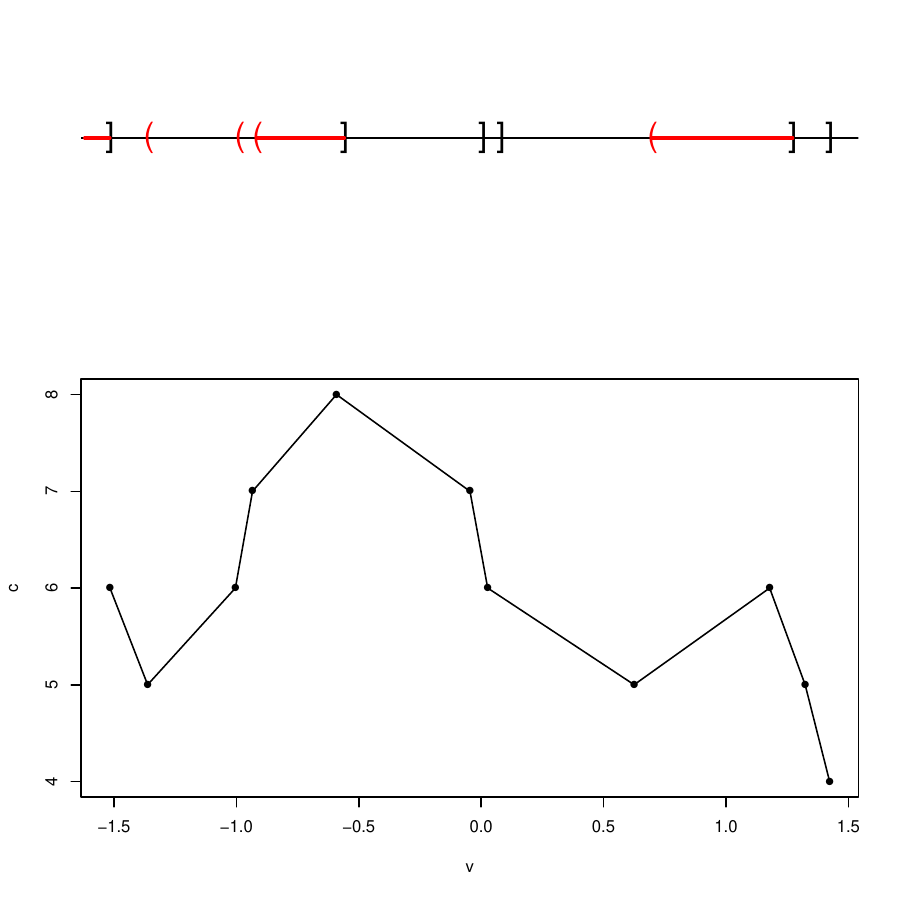}}
    \end{center}
    \caption{\small Intervals for a sample of 10 observations in the univariate model. 
    Values of $v_i$ corresponding to $y_i = 1$ are marked as $]$ in black and those corresponding to $y_i = 0$ are marked as $($ in red.  The three intervals demarcated by the heavier (red) segments are local maxima,  the number of
    counts for each interval is illustrated in the lower portion of the figure.
    }
    \label{fig.Oned}
\end{figure}
The nonparametric maximum likelihood estimator of $F_\eta$ has the following essential features:
\begin{enumerate}
   \item Since the $p_j$'s are assigned to intervals, it does not matter where mass is located 
       within the intervals, as long as it is assigned to a point inside, in this sense the
       NPMLE, $\hat F_\eta$, is defined only on the sigma algebra of sets formed from the intervals, $I_j$. 
       By convention we could assign the mass to the right end of each interval, but we should remember
       that this is only a convention.  The MLE assigns mass to sets, not to points.
   \item Potential non-zero elements of $p$ correspond to intervals $I_j$ with corresponding 
       $c_j$ being a local maximum.  Were this not the case, the likelihood could always be
       increased by transferring mass to adjacent intervals containing larger counts.
       Figure \ref{fig.Oned} plots values for the vector $c$. 
       It is quite efficient to generate the vector $c$ and find these local maxima even 
       when $n$ is very large.
\end{enumerate}

\begin{figure}[bp]
    \begin{center}
    \resizebox{.75\textwidth}{!}{\includegraphics{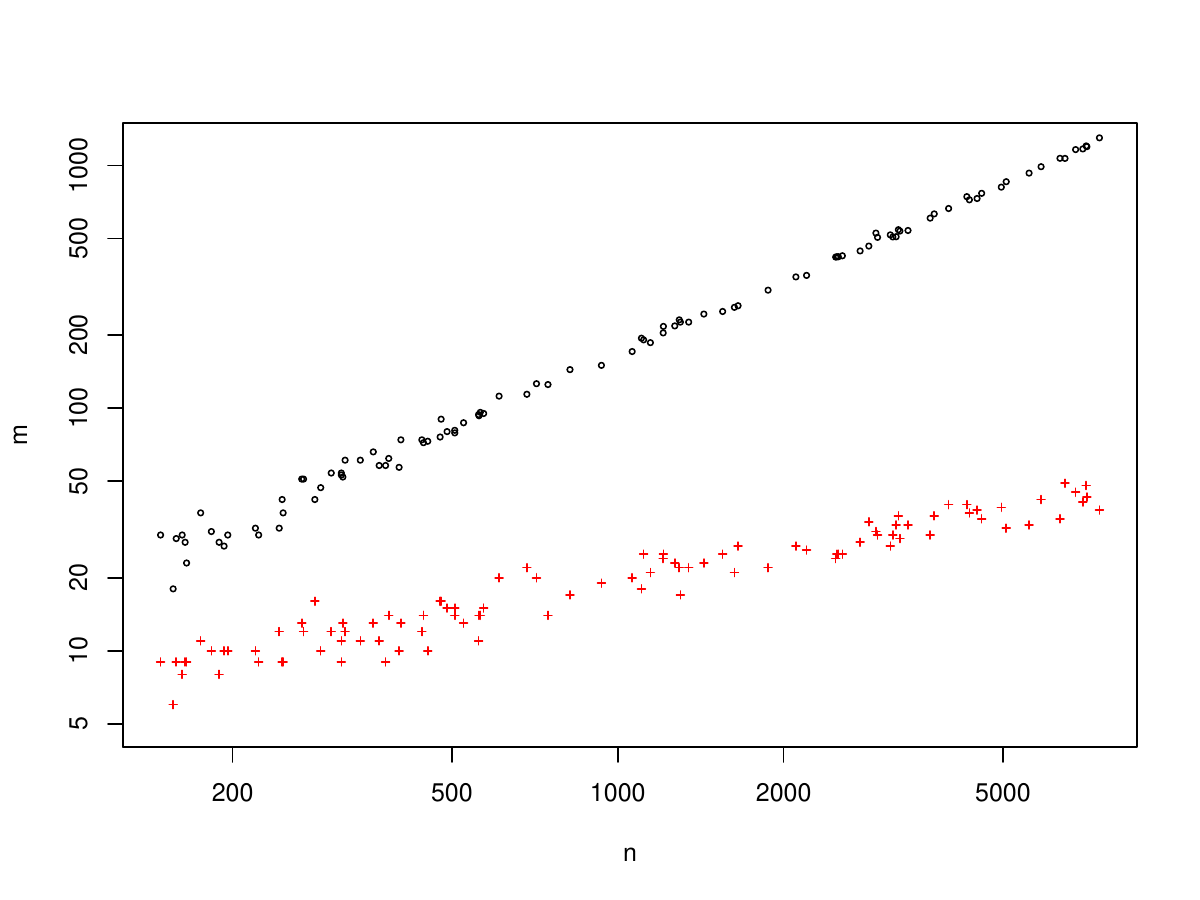}}
    \end{center}
    \caption{\small Number of local maxima of interval counts (black circles) and number 
    of points of support of the optimal NPMLE solution (red pluses).  The local maxima grow
    essentially linearly with sample size, while the number of support points in the NPMLE
    grow roughly logarithmically.
    }
    \label{fig.Fig2a}
\end{figure}

These features suggest a convex optimization strategy for estimation of $F_\eta$.
    Since mass needs only to be assigned to the right endpoint of intervals that correspond to
    local maxima of the counts, we only need consider those intervals and any point within 
    these locally maximal intervals may be considered as potential support points.
    This serves as a dimension reduction device compared to considering all $n$ original data points.
    However, in our experience, the number of potential support points of the distribution $F_\eta$
    identified in this manner still grows linearly with the sample size $n$, while the number of
    optimal support points identified by the maximum likelihood estimator grows more slowly.
    To determine which of our locally maximal intervals deserves positive mass and if so how much, 
    we must solve,
    \begin{equation}
    \underset{f \in \mathcal{S}_{m-1}}{\max}\Big\{  \sum_{i=1}^n \log g_i \mid Af = g\Big\}
    \end{equation}
    where $f = (f_j)$ denotes the mass assigned to the $j$-th order statistic of our reduced set of potential 
    $v$ of cardinality $m$.  The $i$-th coordinate of $Af$ equals to $\sum_j 1\{ v_{(j)} \leq v_i\}f_j$ if $y_i = 1$ and 
    $\sum_j 1\{ v_{(j)} > v_i\}f_j$ if $y_i = 0$.   As first noted by \citeasnoun{Cosslett83}, this problem 
    turns out to be a special instance of the nonparametric maximum likelihood estimator described in \citeasnoun{KW}.
    \citeasnoun{GJ} stress the characterization of the NPMLE in this setting as a convex minorant, 
    thereby linking it to the celebrated Grenander estimator of a monotone density;  the convex optimization
    formulation given here is equivalent as established in \citeasnoun{GW}.
    Noting that the problem is strictly concave in the $g_i$ and subject only to linear equality and
    inequality constraints, it can be solved very efficiently, as noted by \citeasnoun{KM},  
    by interior point methods as implemented for 
    example in Mosek, the optimization framework developed by \citeasnoun{Mosek}.  When this is carried out
    one finds that the number of support points of the estimated distribution, $\hat F_\eta$, is considerably 
    smaller than the number of local maxima identified as candidates.  This is illustrated in Figure 
    \ref{fig.Fig2a}, where  we have generated standard Gaussian $v_i$'s and $\eta_i$'s, for samples
    of size, $n = \exp (\xi)$ with $\xi$ drawn uniformly on the interval $[5,9]$;  
    round black points indicate the number of local
    maxima, while red plus points depict the number of optimal NPMLE support points.  While 
    the number of local maxima grow essentially linearly in the sample size, $n$, the number of positive 
    mass points of the NPMLE grows more slowly, roughly like $\sqrt{n}$. % See Spointsa.R in the figures directory.
    This slow growth in
    the number of mass points selected by the NPMLE is consistent with prior experience with related
    methods for estimating smooth mixture models as described in \citeasnoun{KM} and \citeasnoun{ProbRob}.

    When we admit the possibility of ties in the $v_i$'s increased care is required to correctly
    count the number of observations allocated to each of the intervals.  This is especially true
    when conflicting binary responses are observed at tied values of $v$.  In such cases it is
    convenient to shift locations, $v_i$ of the tied $y_i = 0$ observations slightly thereby
    restoring the uniqueness of the intervals.

\section{Multivariate Randomness}
The convenience of the univariate case is that there is a clear ordering of $v_i$ on $\mathbb{R}$, 
hence the partition is very easy to be characterized and enumerated. This seems to be lost once we 
encounter the multivariate case, however the bivariate current status and bivariate interval censoring
models considered by \citeasnoun{GJ} provide a valuable conceptual transition in which the intervals
of the univariate setting are replaced by rectangles in the bivariate setting.  \citeasnoun{Maathuis}
describes an effective algorithm for these models that shares some features with our approach.

To help visualize the geometry of the NPMLE in our more general setting with multivariate random coefficients 
we will first consider the bivariate case without any auxiliary covariates, $w_i$, and maintain the 
small sample focus of the previous section.  
Now, $x_i = (1, z_i , -v_i )^\top$ where $z_i$ is a random scalar and $\eta_i \sim F_0$. 
The binary response is generated as,
\[
\PP(y_i = 1 | x_i) = \PP(\eta_{1i} + z_i \eta_{2i} - v_i \geq 0 ).
\]
Each pair, $(z_i , v_i )$, defines a line that divides $\RR^2$ into two half-spaces, an ``upper''
one corresponding to realizations of $y_i = 1$, and a ``lower'' one for $y_i = 0$.  Let $R_i$ denote
these half-spaces and $F(R_i)$ be the probability assigned to $R_i$ by the distribution
$F$.  Our objective is to estimate the distribution, $F_0$.  The log likelihood of the 
observed sample is,
\[
\ell(F)= \sum_{i=1}^n \log F( R_i).
\]
As in the univariate case, the $R_i$ partition the domain of $\eta$, however, now rather than
intervals the partition consists of polygons formed by intersections of the $R_i$ half-spaces. 
Adapting our counting method for intervals to these polygons, we seek to identify polygons whose
counts are locally maximal.  Within these maximal polygons the data is uninformative so again 
there is some inherent ambiguity about the nature of the NPMLE solutions.  As in the one dimensional
case this ambiguity can be resolved by adopting a selection rule for choosing a point within each polygon.
As long as there is sufficient variability in the pairs $(z_i , v_i )$'s this ambiguity will vanish as 
the sample size grows as we will consider more formally in Section \ref{sec.asymptopia}.

Figure \ref{fig.maxint} illustrates this partition with $n = 5$. 
%Figure \ref{fig.partition} illustrates this partition with $n = 5$ and $d = 3$. 
The data for this example are given in Table \ref{tab: n=5}. 
 \begin{table}[ht]
	 \centering
	 \begin{tabular}{rrrrr}
		 \hline
		 & (Intercept) & $z_{i}$ & $v_{i}$ & $y_i$ \\ 
		 \hline
		 $i=1$ & 1.00 & 0.41 & 1.22 &1 \\ 
		 $i=2$ & 1.00 & 0.40 & 0.36&0 \\ 
		 $i=3$ & 1.00 & 0.17 & 0.24 &1\\ 
		 $i=4$ & 1.00 & -0.79 & 0.99&0 \\ 
		 $i=5$ & 1.00 & -0.94 & 0.55 &0\\ 
		 \hline
	 \end{tabular}
\caption{Data for a toy example with $n = 5$}
 \label{tab: n=5}
 \end{table}
% The slope and intercept of the boundary of the half-space $R_i$ are determined by 
% $\frac{z_{i}}{-v_{i}}$ and $\frac{1}{-v_{i}}$ and the direction of the normal 
% vector shown in the Figure is determined by the associated $y_i$; the 
% circled number indicates the index $i$. 
% \begin{figure}
	% \label{fig.partition}
	% \includegraphics[scale=0.5]{figures/intersect_n=5plain}
	% \caption{Illustration of half-spaces defined by the data for $n = 5$ and $d=3$.}
% \end{figure}

The half-spaces $\{R_1,\dots, R_n\}$ partition $\RR^2$ into disjoint polygons 
$C_1,\dots, C_M$.  Since for each $i$, it must be that $C_m \cap R_i = \emptyset$ 
or $C_m \cap R_i=C_m$, we can represent 
\[
F( R_i ) = \sum_{j=1}^{M}p_j 1\{C_j \subset R_i\}
\]
where $p_j$ denotes the probability assigned to $C_j$, and we can express the 
log likelihood in terms of these probabilities, $p = (p_j)_{j=1:M}$,   
\[
\ell (p) = \sum_{i=1}^{n} \log \Big(\sum_{j=1}^{M} p_j 1\{C_j \subset R_i\}\Big).
\]
This log likelihood is strictly concave in the $n$-vector of half-space probabilities, 
$q = (F(R_i))$, which is restricted to lie in a convex set and is, therefore, 
uniquely determined.  Uniqueness of the $M$ vector $p$ is, however, a more delicate
issue.  In Appendix \ref{app.nonu} we describe a counterexample for which there is a
convex set of solutions in $p$ each element of which determines the same vector $q$.
The vector $p$ must lie in the $M-1$ dimensional unit simplex and satisfy $q = Ap$ 
for the $n$ by $M$ Boolean matrix, $A = (1(C_j \subset R_i))$. When $M > n$ there is
the possibility that $p$ may be non-unique as noted earlier by \citeasnoun{GG94} and
\citeasnoun{Lindsay.95}.  Just as in the univariate setting
where we need not consider {\it all} intervals only those with locally maximal counts, 
we now need only consider polygons with locally maximal counts,  thereby reducing the dimension, 
$M$, of the vector $p$.  Often, this initial step will reduce $M$ to less than $n$; when 
it fails to do so the non-negativity constraints on the components of $p$ in the convex
optimization of the log likelihood tend to accomplish this.

We now illustrate how to find the maximal polygons for our toy example with $n=5$. 
Figure \ref{fig.maxint} shows that there are three maximal polygons, all shaded and 
denoted as $\{C_1,C_2,C_3\}$. Polygon $C_1$ is the intersection of $\{R_1,R_3,R_4,R_5\}$, 
$C_2$ is the intersection of $\{R_1,R_2,R_4,R_5\}$ and $C_3$ is the intersection of $\{R_1,R_2,R_3\}$. 
Each is locally maximal in the sense that they are formed by more intersecting half-spaces 
than any of their neighbours.  The maximum likelihood estimator for $p$ is defined as any maximizer of,
\[
\underset{p\in \mathcal{S}_2}{\max} \Big\{(p_1+p_2+p_3) (p_2+p_3) (p_1+p_3)(p_1+p_2)(p_1+p_2)\Big\}
\]
which leads to the unique solution, $p_1 = p_2 = 1/2$ and $p_3 = 0$, and the optimal 
log-likelihood of $\log(1/4) = -1.3863$.  We should again stress that although the mass associated
with the two optimal polygons is uniquely determined by maximizing the likelihood, the position
of the mass is ambiguous, confined only to the regions bounded by the two polygons.

\begin{figure}
	\label{fig.maxint}
	\includegraphics[scale=0.5]{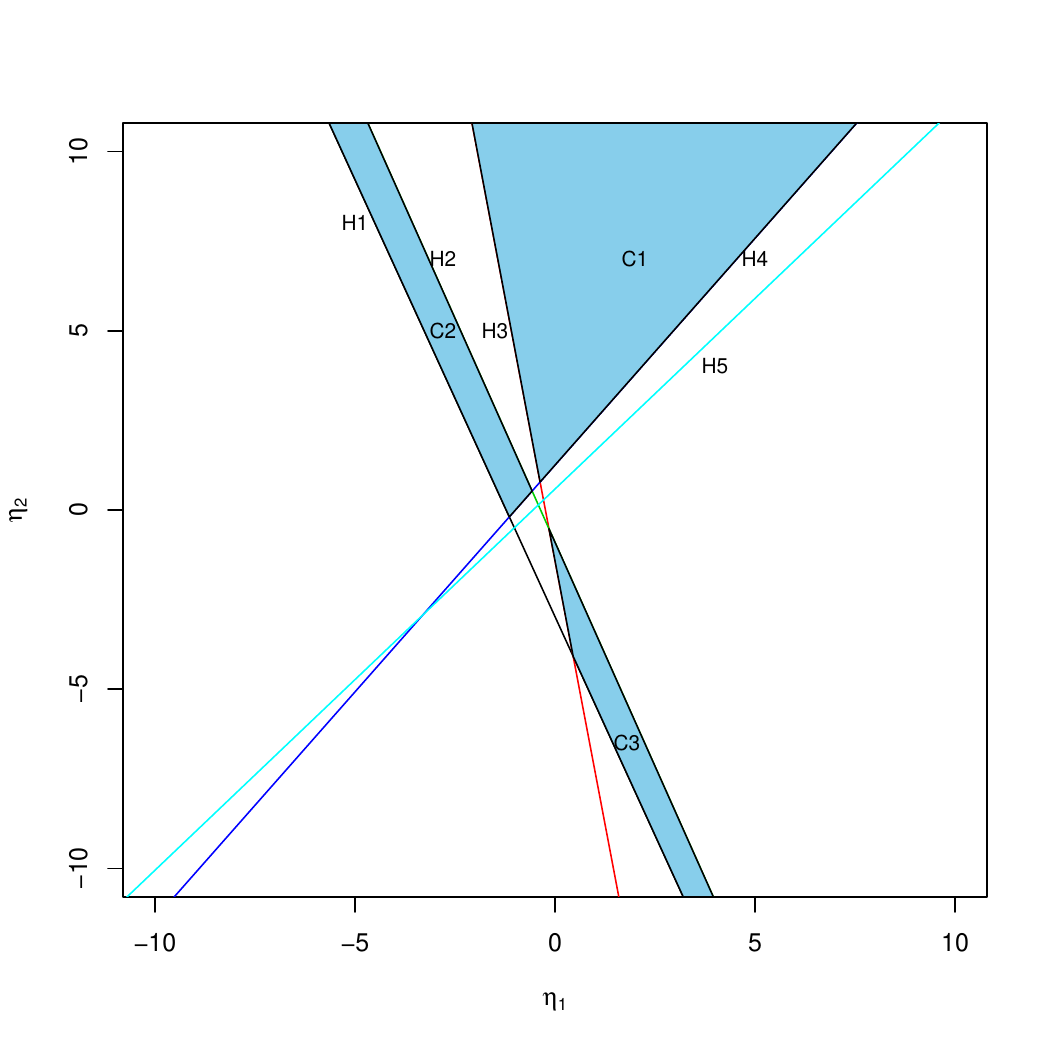}
	\caption{The polygons defined by a toy example with $n = 5$.}
\end{figure}

As $n$ grows the number of maximal polygons grows rapidly and finding all of them in the informal 
manner described above quickly becomes impractical.  Fortunately, there is an extensive, relatively
recent, combinatorial geometry literature on ``hyperplane arrangements'' that allows for efficient and 
tractable enumeration of the partition induced by any set of n hyperplanes in $\mathbb{R}^d$.  

\subsection{Hyperplane Arrangement and Cell Enumeration}
Given a set of hyperplanes $H := \{H_1, H_2, \dots, H_n\}$ in $\mathbb{R}^d$, it defines a partition of 
the space. In the computational geometry literature, this partition is called a hyperplane arrangement 
$\mathcal{A}(H)$.  We first characterize the complexity of an arrangement by the following Lemma,
established by \citeasnoun{Zaslavsky}. 

\begin{lemma} \label{lemma: partition}
	The number of cells of an arrangement of $n$ hyperplanes in $\mathbb{R}^d$ is $\OO(n^{d})$. 
\end{lemma}

\begin{remark}
The worst case occurs when all $n$ hyperplanes are in general position: a hyperplane arrangement
$\HH = \{ H_1, \ldots , H_n \}$ in $\RR^d$ is in general position if for $1 < k \leq  d$
any collection of $k$ of them intersect in a $d - k$ dimensional hyperplane, 
and if $k > d$ any collection of $k$ of them have an empty intersection.  In such cases
the number of cells generated by the $n$ hyperplanes is given by  $\sum_{i=0}^{d} \binom{n}{i}$, 
apparently first proven by \citeasnoun{Buck}.

When $d = 2$ we have lines not hyperplanes.  There are three basic elements in a line arrangement: 
vertices, edges and polygons. Vertices are the zero-dimensional points at which two or more lines intersect. 
Edges are the one-dimensional open line segments or open infinite rays that connect the vertex points. 
Faces, or cells, are the two-dimensional interiors of the bounded or unbounded convex polygons formed
by the arrangement.  If all lines $\{H_1,\dots, H_n\}$ are in general positions, then the number of 
vertices is $\binom{n}{2}$, the number of edges is $n^2$ and the number of cells is 
$\binom{n}{2}+n+1 = \OO(n^2)$. When lines are not in general position, which is likely to occur in 
many empirical settings, the complexity of the line arrangement is characterized in \citeasnoun{AW81}. 
We return to this possibility in Section \ref{sec:Degeneracy} below.
\end{remark}

Our first objective is to enumerate all the polytopes, or cells, formed by a given arrangement, 
denoted as $\{C_1, C_2, \dots, C_M\}$. For each cell $C_j$ we can define a sign vector 
$ s_j \in \{\pm 1\}^{n}$ whose $i$-th element is,
\[
 s_{ij} :=  \begin{cases}
1 & \text{for } Z_i^\top \eta - v_{i}> 0\\
-1 & \text{for } Z_i ^\top \eta -v_{i} < 0,
\end{cases}
\]
where $Z_i := \{1, z_{i}^\top\}^\top$ and $\eta$ is an arbitrary interior point of $C_j$. 
In this form the sign vector ignores the information in $y_i$, but this can easily be 
rectified by flipping the sign of the $i$-th element in the sign vector. 
We will define the modified sign vector to be $\tilde s_j$ with
\[
\tilde s_{ij} := \begin{cases}
1 & \text{for } y_i = 1, Z_i^\top \eta - v_{i} > 0\\
-1 & \text{for } y_i = 1, Z_i^\top \eta - v_{i} < 0\\
1 & \text{for } y_i = 0, Z_i^\top \eta - v_{i} < 0\\
-1 & \text{for } y_i = 0, Z_i^\top \eta - v_{i} > 0.
\end{cases}
\]
Each cell is uniquely identified by its sign vector and an associated interior point $\eta$. 
The interior point $\eta$ is arbitrary, but since the likelihood is determined only by the probability 
mass assigned to each cell, we need only find a valid interior point $\eta_j$ for each polytope $C_j$. 
This can be accomplished by examining solutions to the linear program,
\begin{equation}
\{\eta_j^*, \epsilon_j^*\} = \underset{\eta,\epsilon}{\argmax} \Big\{ \epsilon \; | \; S_j ( Z \eta - v) 
\geq \epsilon \mathbf{1}_n, 0 \leq \epsilon \leq 1 \Big\}.
 \label{eq: LP}
\end{equation}
Here, $S_j$ is a $n \times n$ diagonal matrix with diagonal elements $\{s_{1j}, s_{2j}, \dots, s_{nj}\}$. 
It can be seen that $\eta_j$ is a valid interior point of $C_j$ if and only if $\epsilon_j^* > 0$. 
The upper bound for $\epsilon$ can be changed to any arbitrary 
positive number, thereby influencing  which interior point found as the optimal solution. 

The linear program (\ref{eq: LP}) thus provides a means to check if a particular configuration of the 
sign vector is compatible with a given hyperplane arrangement. A brute force way to enumerate all 
cells in the arrangement is to exhaust all possible sign vectors in the set $\{\pm 1\}^n$, of
which there are $2^n$ elements.
For each element we could solve the linear programming problem and check the existence of a valid 
interior point. This is obviously computationally ridiculous and unnecessary since the maximum number 
of cells generated by a $n$-hyperplane arrangement in $\mathbb{R}^d$ is only of order $n^d$ as stated in Lemma \ref{lemma: partition}.
\citeasnoun{AvisFukuda} were apparently the first to develop an algorithm for cell enumeration 
that runs in time proportional to the maximum number of polygons of an arrangement. 
\citeasnoun{Sleumer} improved upon their reverse search algorithm. More recently, 
\citeasnoun{RadaCerny} have proposesd an incremental enumeration algorithm that is asymptotically 
equivalent to the Avis-Fukuda's reverse search algorithm, but is demonstrably faster 
in finite samples. The most costly component of the Rada-\v{C}ern\'y algorithm involves solving the 
linear programs (\ref{eq: LP}). We  will briefly describe the Rada-\v{C}ern\'y Incremental Enumeration (IE)
algorithm and then discuss a modified version that we have developed that reduces the complexity 
of IE algorithm by an order of magnitude $n$. 

As the name suggests the IE algorithm adds hyperplanes one at a time to enumerate all sign vectors 
of an arrangement in $n$ iterations.  Let $s^k$ denote a sign vector of length 
$k$ in the set of possible vectors  $\{\pm 1\}^k$. In the $k$-th step of the algorithm with 
$1<  k \leq n$, we have as input the sign vectors $s^{k-1}$ for all existing cells formed by 
the first $k-1$ hyperplanes and their associated interior points collected in the set $\eta^{k-1}$;
we will index its elements by $\ell$.  At each iteration we undertake Algorithm \ref{alg:IE}.
\begin{algorithm}[h]
	\SetKwInOut{Input}{input}\SetKwInOut{Output}{output}
	\Input{The existing sets of sign vectors $s^{k-1}$ and interior points $\eta^{k-1}$ and 
	the new hyperplane $H_k$ and associated covariate vector $Z_k$.}
	\Output{The new sets of sign vectors $s^k$ and the interior points $\eta^k$. }
	\begin{enumerate}
		\item[(a)] For each elements in the set $s^{k-1}$, say $s_\ell^{k-1}$, 
		    define the new sign vector $s_\ell^k:= \{s_{\ell}^{k-1}, 
		    2\times 1\{Z_k^\top \eta_{\ell}^{k-1} - v_k> 0\}-1\}\}$. 
		\item[(b)] For each element $s_\ell^k$ obtained from (a), solve the linear program 
		    defined in (\ref{eq: LP})  with the sign vector 
		    $\check{s}_\ell^k := \{s_{\ell}^{k-1}, 1 - 2\times 1\{Z_k^\top \eta_{\ell}^{k-1} - v_k> 0\}\}$, 
		    store the optimal solutions $\check{\eta}_{\ell}^k$ and $\check{\epsilon}_\ell^{k}$.
		\item[(c)] Define the set $s^k:= \{ s_\ell^{k-1}, \forall \ell \} \cup \{\check{s}_\ell^k, \forall \ell \text{ such that }  \check{\epsilon}_\ell^{k}> 0\}$ and the set $\eta^k:= \{\eta_\ell^{k-1}, \forall \ell\} \cup \{\check{\eta}_\ell^k, \forall \ell \text{ such that }  \check{\epsilon}_\ell^{k}> 0\}$.
	\end{enumerate}
	\caption{Iteration of the Incremental Enumeration algorithm}\label{alg:IE}
\end{algorithm}

The algorithm can be initiated with an arbitrary point $\eta_1^1$, as long as it does not fall 
exactly on any of the hyperplane. When the first hyperplane $H_1$ is added, it necessarily partitions 
the space $\mathbb{R}^d$ into two parts. Since $\eta_1^1$ has to belong to one of the parts, 
we define $s_1^1 = 2 \times 1\{Z_1^\top \eta_1^1 - v_1 > 0\} - 1$. For the other half, one solves 
the linear program (\ref{eq: LP}) with $S = s_2^1 = -s_1^1$ and records the solution as 
$\eta_2^1$; defining $s^1 := \{s_1^1, s_2^1\}$ and $\eta^1 := \{\eta_1^1, \eta_2^1\}$ 
which become the input for the first iteration of the IE algorithm. The order of addition of the 
hyperplanes does not influence the final output. 

Figure \ref{fig:nicer} illustrates the idea of the iterative algorithm for a linear arrangement. The input at the 
third iteration is illustrated in the left panel. There are four polygons that partition 
$\mathbb{R}^2$ determined by the arrangement $\{H_1, H_2\}$. The interior points are labeled 
by their associated sign vector represented by the symbols $\pm$. When $H_3$ is added, 
step one  of the algorithm determines $s_{\ell}^3$ for $\ell = 1, \dots, 4$. 
Step two looks for the new polygons created by adding $H_3$, which crosses through polygons 
A,C,D, dividing each into two parts, and leads to the three new polygons E,F and G. 
Polygon B lies strictly on one side of $H_3$, hence it is not divided. 

\begin{figure}
	\centering
	\includegraphics[scale = 0.65]{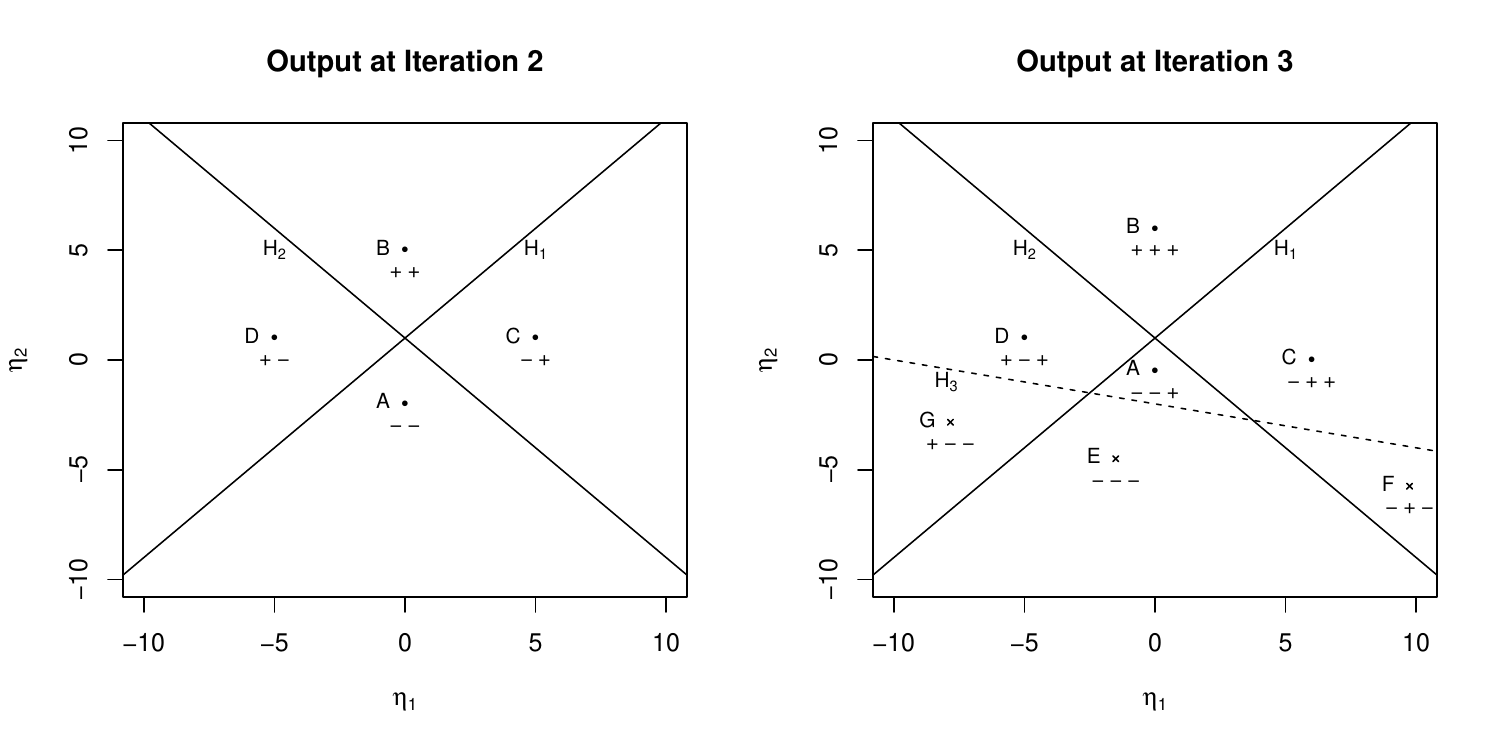}
	\label{fig:nicer}
	\caption{Illustration of the IE algorithm: the left panel plots the output of the second iteration. 
	The two hyperplane involves are $H_1= \{(\eta_1,\eta_2)\mid -\eta_1 + \eta_2 -1 = 0\}$ and 
	$H_2= \{(\eta_1,\eta_2)\mid \eta_1 + \eta_2 -1 = 0\}$. The right panel plots the output of the third iteration, 
	the additional hyperplane $H_3= \{(\eta_1,\eta_2)\mid 0.2 \eta_1 + \eta_2 +2 = 0\}$.}
\end{figure}

The most time-consuming step is (b) in each iteration, although each LP problem can be very quickly 
solved as long as the dimension $d$ is moderate. 
We have to solve $\sum_{i=0}^d\binom{k-1}{i} = \OO((k-1)^d)$ such problems in the worst case as a 
consequence of Lemma \ref{lemma: partition}. 
Together, for all $n$ iterations, this requires $\OO (n^{d+1})$ for any arrangement in $\mathbb{R}^d$. 
%which can be overwhelming when $n$ is large. [let's not advertise this!]

\subsection{Dimension reduction based on locally maximal polygons}
Once we have enumerated all the cells, an adjacency matrix $A$ can be created with dimension 
$n \times M$ and elements taking values in $\{0,1\}$. The $j$-th column of the A matrix is the 
corresponding modified sign vector $\tilde s_j$ of cell $C_j$ to incorporate the information
from $y$, except that we replace all $-1$ values
by $0$.  As we have already noted the number of columns of the A matrix, $M$ is of order $n^d$ which 
increases rapidly as $n$ increases,  so it is important to try to reduce the number of candidate cells
as much as possible.  
To achieve a dimension reduction we need to eliminate cells that are not locally maximal, 
that is cells that have neighbours with larger cell counts. 
Define the $M$-vector $c$ whose elements denote the column sums of the adjacency matrix $A$,
corresponding to each cell $\{C_1, \dots, C_M\}$. 
The cells $C_j$ and $C_k$ are neighbours if their sign vector differ by exactly one sign. 
For each cell $C_j$ we can define its set of neighbours $N_j$. 
The cardinality of any set $N_j$ is at most $n$. 
The following result determines the set of columns in the matrix $A$ that will be locally maximal
and therefore constitute candidate supporting cells of the NPMLE. 

\begin{theorem}\label{thm:N1}
Any cell $C_j$ with associated count, $c_j$,  that is strictly smaller than any of the 
counts of the cells in its neighbouring  set $N_j$ is assigned  zero probability mass by the NPMLE.
\end{theorem}

\begin{remark}
    It is perhaps worth noting at this point that the cell, or cells, possessing the globally maximal
    cell count constitute an exhaustive solution to the maximum score problem posed in \citeasnoun{Manski75},
    that is, any element of the set $\{ C_k : k \in \KK \}$ for 
    $\KK = \{ k : c_k = \max \{ c_j : j = 1, \ldots , M\}\}$ 
    is an argmax of the maximum score objective function.   Likewise, any other locally maximal cell is
    a region in which the maximal score estimator may become marooned in the search for a global maximum.
\end{remark}

\subsection{Acceleration of the Rada-\v{C}ern\'y algorithm}
The main motivation of our modification of the IE algorithm is to reduce the number of LP problems 
that need to be solved. To this end we apply a ``zone theorem'' for hyperplane arrangements 
of \citeasnoun{Edelsbrunner93} to show that the number of necessary LP problems in step (b) 
of the IE algorithm can be reduced from $\OO(n^{d+1})$ to $\OO(n^d)$. 
\begin{theorem}\label{thm: zone}
	For any set of $H$ of $n$ hyperplanes in $\mathbb{R}^d$ and any hyperplane $H' \notin H$, denote the total number of cells in the arrangement of $\mathcal{A}(H)$ that intersects with $H'$ as $h$, then 
	\[
	h \leq \sum_{i=0}^{d-1} \binom{n}{i}
	\]
	with the equality achieved when all hyperplanes in the set $H\cup \{H'\}$ are in general position. 

\end{theorem}

Theorem \ref{thm: zone} implies that when we add a new distinct hyperplane to an existing arrangement 
with $k-1$ hyperplanes, it crosses at most $\OO((k-1)^{d-1})$ cells. For line arrangements this 
implies that a newly added distinct line crosses at most $k$ polygons. Only these $k$ polygons will 
generate a new cell, hence in step (b) of the $k$-th iteration of the IE algorithm, we need only 
to solve at most $k$ LPs, in contrast to $k +\binom{k-1}{2}$ provided we can efficiently find the relevant 
$k$ sign vectors for these crossed cells. We first give details for line arrangement in Algorithm \ref{alg:AIE} 
and then discuss how to extend this approach to general case in $\mathbb{R}^d$.

\begin{algorithm}
\SetKwInOut{Input}{input}\SetKwInOut{Output}{output}
\Input{The existing set of sign vectors $s^{k-1}$ and interior points $\eta^{k-1}$ and 
    the new line $H_k$ and associated covariate vector $Z_k$.}
\Output{The new set of sign vectors $s^k$ and interior points $\eta^k$. }
    \begin{enumerate}
	\item[(a)] For each element in the set $s^{k-1}$, say $s_\ell^{k-1}$, define the new sign 
	    vector $s_\ell^k:= \{s_{\ell}^{k-1}, 2\times 1\{Z_k^\top \eta_{\ell}^{k-1} - v_k> 0\}-1\}\}$. 
	\item[(b.1)] Find the set of vertices $\{t_1, t_2, \dots, t_{k-1}\}$, where $t_j$ is the 
	    vertex of the intersection of $H_k$ and $H_j$ for $1 \leq j <k$. 
	    Let $u_j^{k-1} := \{sgn\{Z_i^\top t_j - v_i\}\}_{i = 1, 2, \dots, k-1}$. 
	    By the definition of a vertex, the $j$-th entry of $u_j^{k-1}$ is zero and the 
	    remainder take values in $\{+1, -1\}$. 
	    Define $u_{j+}^{k-1}$ to be identical to $u_{j}^{k-1}$ except that its $j$-th entry 
	    is replaced by $+1$ and $u_{j-}^{k-1}$ to be identical to $u_{j}^{k-1}$ except that its $j$-th 
	    entry is replaced by $-1$. Let $u^{k-1} := \{u_{j-}^{k-1}, u_{j+}^{k-1}\}_{j=1,2,\dots, k-1}$ 
	    and $\ZZ^{k-1} := u^{k-1}\cap s^{k-1}$ and denote the corresponding set of interior points 
	    as $\check{\eta}^{k-1}$. 
	\item [(b.2)] For each element in the set $\ZZ^{k-1}$, say $\ZZ_{\ell}^{k-1}$, solve 
	    the linear program defined in (\ref{eq: LP}) with the sign vector 
	    $\tilde{s}_{\ell}^k := \{\ZZ_\ell^{k-1}, 1 - 2 \times 1\{Z_k^\top \check{\eta}_{\ell}^{k-1} - v_k > 0\}\}$,
	    store the optimal solutions $\tilde{\eta}_{\ell}^k$. 
	\item[(c)] Define the set $s^k:= \{ s_\ell^k, \forall \ell \} \cup \{ \tilde{s}_\ell^k, \forall \ell\}$ 
	    and the set $\eta^k:= \{\eta_\ell^{k-1}, \forall \ell\} \cup \{\tilde{\eta}_\ell^k, \forall \ell \}$.
	\end{enumerate}
	\caption{Accelerated Incremental Enumeration algorithm (d = 2)}\label{alg:AIE}
\end{algorithm}

The Accellerated Incremental Enumeration (AIE) Algorithm \ref{alg:AIE} describes how to achieve this when all lines 
are in general position.  In each iteration, Step (b.1) finds the subset of sign vectors in $s^{k-1}$ whose corresponding 
cell  intersects with the new line $H_k$. Step (b.2) then finds the interior points for all the newly 
created cells by solving the associated LPs.  When lines are not in general position (i.e. more than two 
lines cross at the same vertex) more entries in the 
vector $u_j^{k-1}$ will be zero, and the set $\ZZ^{k-1}$ can be constructed in a similar fashion,
as noted in the next subsection.

The AIE algorithm is easily adapted to the  general case  with hyperplanes in $\mathbb{R}^d$, at least when
the arrangement is in general position.  When a new hyperplane $H_k$ is added, the set of vertices is determined 
by the intersection of $H_k$ and any $d-1$ hyperplanes in the set $\{H_1, H_2, \dots, H_{k-1}\}$. 
When hyperplanes are in general positions, there will be $\binom{k-1}{d-1}$ of these for $k \geq d$. 
Each of these vertices provide sign constraints on the cells that hyperplane $H_k$ crosses, 
which allows us to determine a subset of $s^{k-1}$ to be passed into step (b.2). 
For $1<k<d$ the set of vertices is empty and we proceed to step (b.2) to process all elements in $s^{k-1}$. 

\subsection{Treatment of various forms of degeneracy} \label{sec:Degeneracy}
If the arrangement $H$ is not in general position, Algorithm \ref{alg:AIE} must be adapted to cope with this.
\citeasnoun{AW81} consider cell enumeration for $d = 2$ and $d = 3$, our implementation of the algorithm 
provisionally treats only the $d = 2$ case of line arrangements.  In higher dimensions, degeneracy becomes
quite delicate and constitutes a subject for future research.  Random perturbation of the covariate data
is an obvious alternative strategy for circumventing such degeneracy.  Even for line arrangements there are several 
cases to consider:
\begin{enumerate}
	\item If $H_k \in \{H_1,\dots, H_{k-1}\}$ for any $k$, then (b.1) and (b.2) can be skipped 
	    since if the new line coincides with one of the existing ones then no new cells are created.
	\item If $H_k$ is parallel to any existing lines $\{H_1, \dots, H_{k-1} \}$ the cardinality of 
	    the set of vertices will be smaller than $k-1$, but no modification of the algorithm is required.
	\item If $H_k$ crosses a vertex that already exists, then more than one entry in $u_j^{k-1}$ will be zero. 
	    Suppose there are $a$ such zeros, then the zero entries must be replaced with elements in $\{\pm1\}^a$ 
	    to construct the set $u^{k-1}$. 
\end{enumerate}

\section{Identification and Strong Consistency}\label{sec.asymptopia}

We now introduce formal conditions needed for identification of the parameters of interest 
$(\theta_0, F_0)$ for model (\ref{eq: model}) and to establish consistency of the nonparametric maximum likelihood 
estimator $(\hat \theta_n, \hat F_n)$. As discussed earlier, identification in binary choice random coefficient 
model requires a normalization and to this end we assume that the coefficient of the last variable in 
$x_i = \{1, z_i^\top, -v_i\}^\top$ has coefficient 1 and therefore the model becomes 
$y_i = 1\{w_i ^\top \theta_0 +\beta_{1i}+ z_i ^\top \beta_{-1i} \geq v_i\}$. 
Here $\beta_{1i}$ refers to the first element in $\beta_i$ and the rest of the vector $\beta_i$ 
is denoted as $\beta_{-1i}$. Note that the sign of the coefficient of $v_i$ is identifiable from the data.

\begin{assumption}\label{Assump1}
	The random vectors $(Z, V, W)$ and $\beta_i$ are independent. 
\end{assumption}

\begin{assumption}\label{Assump2}
	The parameter space $\Theta$ is a compact subset of a Euclidean space and $\theta_0 \in \Theta$. 
	Let the set $\mathcal{F}$ be space of probability distributions for $\beta_i$ supported on a 
	compact set in $\mathbb{R}^d$. 
\end{assumption}

\begin{assumption}\label{Assump3}
The random variable $V$ conditional on $(W, Z)$ is absolutely continuous and has full support on $\mathbb{R}$ and the random variables $Z$ conditional on $W$ are absolutely continuous and has full support on $\mathbb{R}^d$. 
\end{assumption}

Versions of the foregoing assumptions are commonly invoked in the semiparametric single index binary choice 
model literature, e.g. \citeasnoun{KleinSpady}, \citeasnoun{Ichimura}. Some relaxation of 
Assumption \ref{Assump1} is possible while still securing point 
identification of $(\theta_0, F_0)$, as noted in the remark following Theorem \ref{thm: identification}. 
The bounded support assumption on $\beta_i$ in Assumption \ref{Assump2} is not needed for 
identification but is convenient for the consistency proof. It can be relaxed at the cost of a slightly longer proof. 
Assumption \ref{Assump3} is the most crucial for the identification argument that requires sufficient 
variability of the covariates to trace out $F_0$ on its full support. Note that it does not require all elements 
in $w_i$ to have an absolutely continuous distribution with full support, in fact $w_i$ can even contain discrete 
covariates.

\begin{theorem}\label{thm: identification}
	Under Assumptions \ref{Assump1}-\ref{Assump3}, $(\theta_0, F_0)$ is identified. 
\end{theorem}

\begin{remark}
The full support requirement on $Z$ in Assumption \ref{Assump3} may be relaxed at the 
cost of imposing tail conditions on the distribution of the random coefficients.  See
\citeasnoun{masten2017random} for further details.
If elements of $z_i$ are endogenous but there exists a vector of instruments $r_i$ and a 
complete model relating $z_i$ and $r_i$ is specified, for instance 
$z_i = \Psi r_i  + e_i$, then we can rewrite model (\ref{eq: model}) as 
\[
y_i = 1\{ \beta_{1i} + e_i ^\top \beta_{-1,i} +  r_i^\top \Psi^\top \beta_{-1,i} + w_i ^\top \theta \geq v_i\} 
\]
Thus, we can redefine the random intercept as $\beta_{1i} + e_i ^\top \beta_{-1,i}$ and the remaining
random coefficients accordingly, denoting by $\beta_{-1,i}$ the vector of $\beta$ excluding the first 
component, and we are back to the original model (\ref{eq: model}). There is also a recent literature 
emphasizing on set identification of $(\theta_0, F_0)$ where an explicit model between $z_i$ 
and $r_i$ is not imposed. See \citeasnoun{ChesherRosen14} for a detailed discussion. 

\end{remark}

Having established identification of the model structure $(\theta_0,  F_0)$, we now turn 
our attention to the asymptotic behavior of the maximum likelihood estimator of $(\theta_0,  F_0)$,
\[
(\hat \theta_n, \hat F_n) = \underset{\Theta \times \mathcal{F}}
{\argmax} \frac{1}{n} \sum_{i=1}^{n} y_i \log [ \mathbb{P}_{F}(H(x_i, w_i, \theta))] + 
(1-y_i) \log [1-\mathbb{P}_{F}(H(x_i, w_i, \theta))]
\]
where we denote the set $\{\beta:  x_i^\top \beta + w_i^\top \theta \geq 0\}$ by $H(x_i, w_i, \theta)$. 
For any fixed $n$ and $\theta$, the collection of half-spaces $H(x_i, w_i, \theta)$ defines a partition on 
the support of $\beta$, denoted as $\{C_1, \dots, C_M\}$. Define a matrix $A$ of dimension 
$n \times M$, whose entries takes the form $a_{ij} = 1\{C_j \subset 
H(x_i, w_i, \theta)\}$ for $y_i = 1$ and $a_{ij} = 1- 1\{C_j \subset H(x_i, w_i, \theta)\}$ for $y_i = 0$. 
Let $p$ be a vector in the unit simplex such that $p_j$ corresponds to the probability assigned to the 
polytope $C_j$, the maximum likelihood estimator for $F$ for a fixed $\theta$, denoted as $F_{\theta}$, 
is the solution to the following constrained optimization problem, 
\[
\min \Big \{- \frac{1}{n} \sum_{i=1}^n \log g_i \mid g_i = \sum_j a_{ij} p_j , \sum_j p_j = 1, p_j \geq 0 \Big \}
\]
The dual problem, which is usually more efficiently solved since $M$ is typically much larger than $n$, 
can be shown to be 
\[
\max \Big \{ \sum_{i=1}^n \log q_i \mid \sum_{i=1}^n a_{ij} q_i \leq n \text{ for all j } \Big\}
\]
In this formulation the problem appears to resemble the dual form of the NPMLE for general mixture
models as considered by \citeasnoun{Lindsay83} and \citeasnoun{KM}.
Despite this resemblance, there are several fundamental differences that lead to special features 
of the binary response NPMLE.  First, in classical mixture models the number of constraints in the 
dual problem is typically infinite dimensional.  As a consequence it is conventional to
impose a finite grid for the potential mass points of the mixing distribution to obtain a 
computationally practical approximation.  The number of grid points should typically grow together 
with sample size $n$ to achieve a good approximation.  In contrast, in the binary response setting 
once $n$ is fixed, the matrix $A$ is determined by the arrangement and the solution of $p$ is exact for 
any $n$. There is no need to impose a grid on the support of the parameters which is very convenient 
especially when the dimension of $\beta$ is large.  Second, the arrangement also provides a unique 
geometric underpinning  which leads to an equivalence class of solutions for the maximum likelihood estimator. 
Once the convex program determines the optimal allocation $\hat p$, the data offers no information on 
how the probability masses $\hat p_j$, should be distributed on the polytope $C_j$.  In this sense the
NPMLE $\hat F_n$ yields a set-valued solution; in the application section we illustrate the implications
of this fact for prediction of marginal effects.  The maximum likelihood estimator for $\theta$ is 
found by maximizing the profile likelihood.

We now establish the consistency of the maximum likelihood estimator. The argument follows 
\citeasnoun{KW} and \citeasnoun{Chen17}. Our argument relaxes some of the assumptions employed in prior work, 
in particular Assumption 4, in \citeasnoun{Ichimura}.  
Let $\mathcal{F}_0$ be the set of all absolutely continuous distribution on the support of $\beta$;
$\mathcal{F}_0$ is a dense subset of $\mathcal{F}$. 

\begin{lemma}\label{lemma: continuity}
Under Assumption \ref{Assump1}-\ref{Assump3}, for any given $\gamma = (\theta, F) \in \Theta \times \mathcal{F}_0$, 
let $\gamma_n = (\theta_n, F_n)$ be any sequence in $\Theta \times \mathcal{F}$ such that $\gamma_n \to \gamma$, then 
\[
\underset{\gamma_n \to \gamma}{\lim} \mathbb{P}_{F_n}(H(x, w, \theta_n)) = 
\mathbb{P}_{F}(H(x, w, \theta))\quad  \text{a.s.}
\]
\end{lemma}

Define the function $p(y, x, w, \gamma) := y \mathbb{P}_F(H(x,w, \theta)) + (1-y) (1-\mathbb{P}_F(H(x,w, \theta)))$,
and for any subset $\Gamma \subset \Theta \times \mathcal{F}$, let $p(y, x, w, \Gamma) = 
\underset{\gamma \in \Gamma}{\sup~}  p(y, x, w, \gamma)$. Let $\rho$ be a distance on 
$\Theta \times \mathcal{F}$.  This $\rho$ can be the Euclidean distance on $\Theta$ together with any 
metric on $F$ that metrises weak convergence of $F$. 
For any $\epsilon>0$ let $\Gamma_{\epsilon}(\gamma^*) = \{\gamma: \rho(\gamma, \gamma^*) < \epsilon\}$ be an 
open ball of radius $\epsilon$ centered at $\gamma^*$. 

\begin{lemma}\label{lemma: integrability}
Under Assumption \ref{Assump1}-\ref{Assump3}, for any $\gamma \neq \gamma^*$, there exists 
an $\epsilon>0$ such that 
	\[
	\mathbb{E^*}\{[\log \{p(y,x, w, \Gamma_{\epsilon}(\gamma))/p(y, x, w, \gamma^*)\}]^+\} < \infty
	\]
	where $\mathbb{E}^*$ denotes expectations taken with respect to $p(y,x, w, \gamma^*)$. 
\end{lemma}

\begin{theorem}\label{thm: consistency}
    If $\{(y_i, x_i, w_i): i = 1, 2, \ldots , n \}$ is an iid sample from $p(y, x, w, \theta_0,  F_0)$, 
	under Assumptions \ref{Assump1}-\ref{Assump3}, then $(\hat \theta_n, \hat F_n)$ is strongly consistent. 
\end{theorem}

\section{Some Simulation Evidence}

In this section we report on some very limited  simulation experiments designed to compare performance
of our NPMLE method with the recent proposal by \citeasnoun{GK}.   The Gautier and Kitamura estimator may be
viewed as a deconvolution procedure defined on a hemisphere in $\RR^{d}$, and has the notable virtue
that it can be computed in closed form via elegant Fourier-Laplace inversion formulae.  A downside of
the approach is that it involves several tuning/truncation parameters that seem difficult to select.  In
contrast the optimization required by our NPMLE  is tuning parameter free, 
and the likelihood interpretation of the resulting convex optimization problem offers the
opportunity to formulate extended versions of the problem containing additional fixed parameters that
may be estimated by conventional profile likelihood methods.  

To facilitate the comparison with \citeasnoun{GK} we begin by considering the two simulation settings adapted
from their paper.  Rows of the design matrix, $X$ are generated as $(1, x_{i1}, x_{i2})$ with standard
Gaussian $x_{ij}$, and then normalized to have unit length.  The random coefficients, $\beta$
are generated  in the first setting from the two point distribution that puts equal mass on the points,
$(0.7,-0.7,1)$ and $(-0.7,0.7,1)$.  This is a highly stylized variant of the Gautier-Kitamura setting 
intended to be favorable to the NPMLE.  The sample size is 500.  Estimation imposes the condition that the third 
coordinate of $\beta$ is 1, and interest focuses on estimating the distribution of the first two coordinates.
After some experimentation with the authors' Matlab code we have implemented a version of the
\citeasnoun{GK} estimator in R.  We initially set the truncation and trimming parameters for
the estimator as suggested in \citeasnoun{GK}, and we plot contours of the resulting density
estimator in left panel of Figure \ref{fig.GK0}.  The discrete mass points of the data generating process
are depicted as red circles in this figure.  It may be noted that the estimated  contours
tend to concentrate the mass too much toward the origin.  In an attempt to correct this bias
effect, we experimented with increasing the truncation parameters to increase the flexibility
of the sieve expansion.  The right panel of Figure \ref{fig.GK0} illustrates the contours of the fit with $T = 3$
replaced by $T = 7$.  The two most prominent modes are now much closer to the discrete mass points
of the process generating the data, but there is some cost in terms of increased variability.
In Figure \ref{fig.KW0}  we illustrate estimated mass points as well as contours of the NPMLE of the random 
coefficient density after convolution with a bivariate Gaussian distribution with diagonal covariance
matrix with entries $(0.04, 0.04)$.  The NPMLE is extremely accurate in this, somewhat artificial
discrete setting.

\begin{figure}
	\includegraphics[scale=0.6]{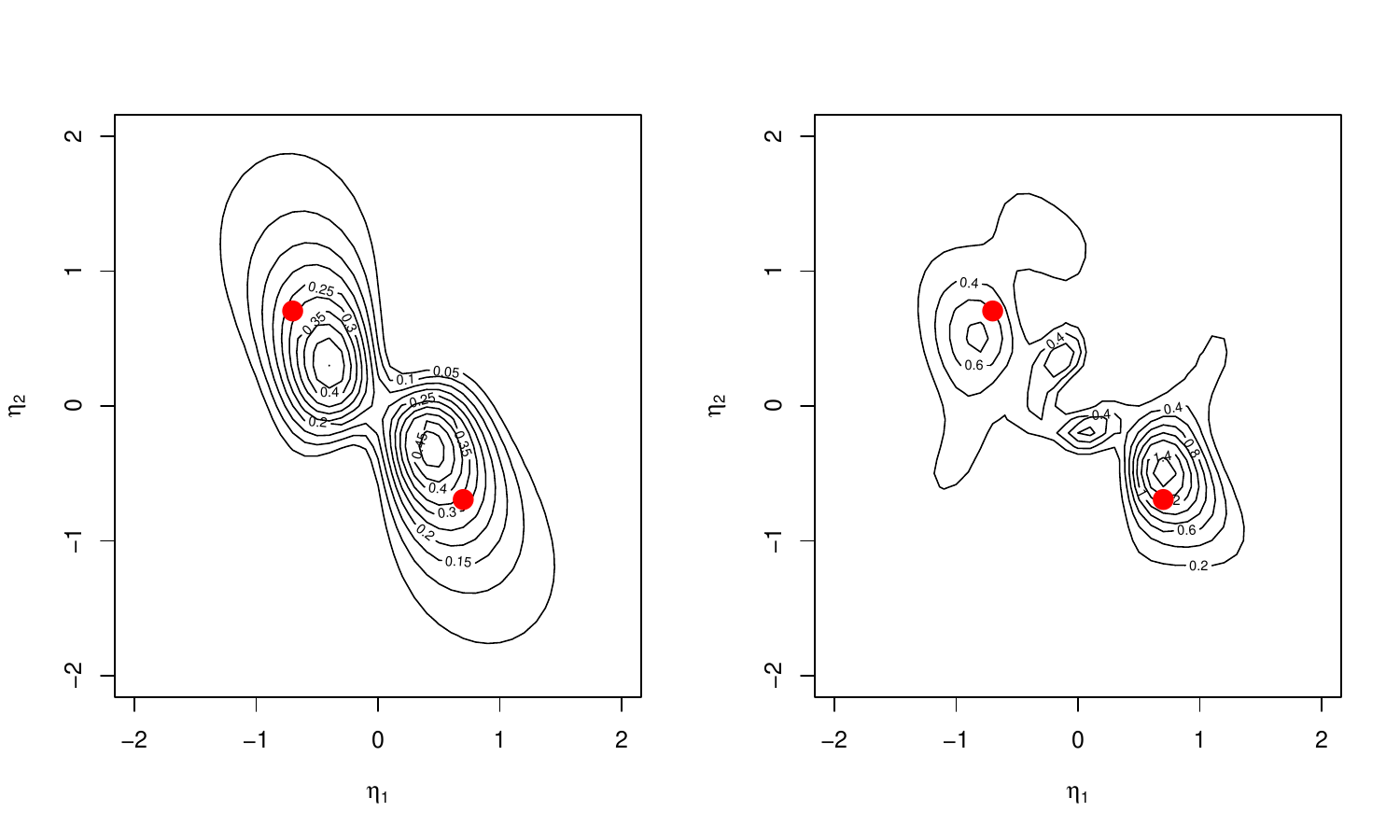}
	\caption{Contour plots of Gautier and Kitamura estimated density for the discrete simulation
	setting with $\eta$ generated with equal probability from the two points
	$(0.7,-0.7)$ and $(-0.7,0.7)$ indicated by the red circles.
	In the left panel the sieve dimension is set at the default value $T = 3$, while in the right
	panel it is increased to $T = 7$.}
	\label{fig.GK0}
\end{figure}
\begin{figure}
	\includegraphics[scale=0.6]{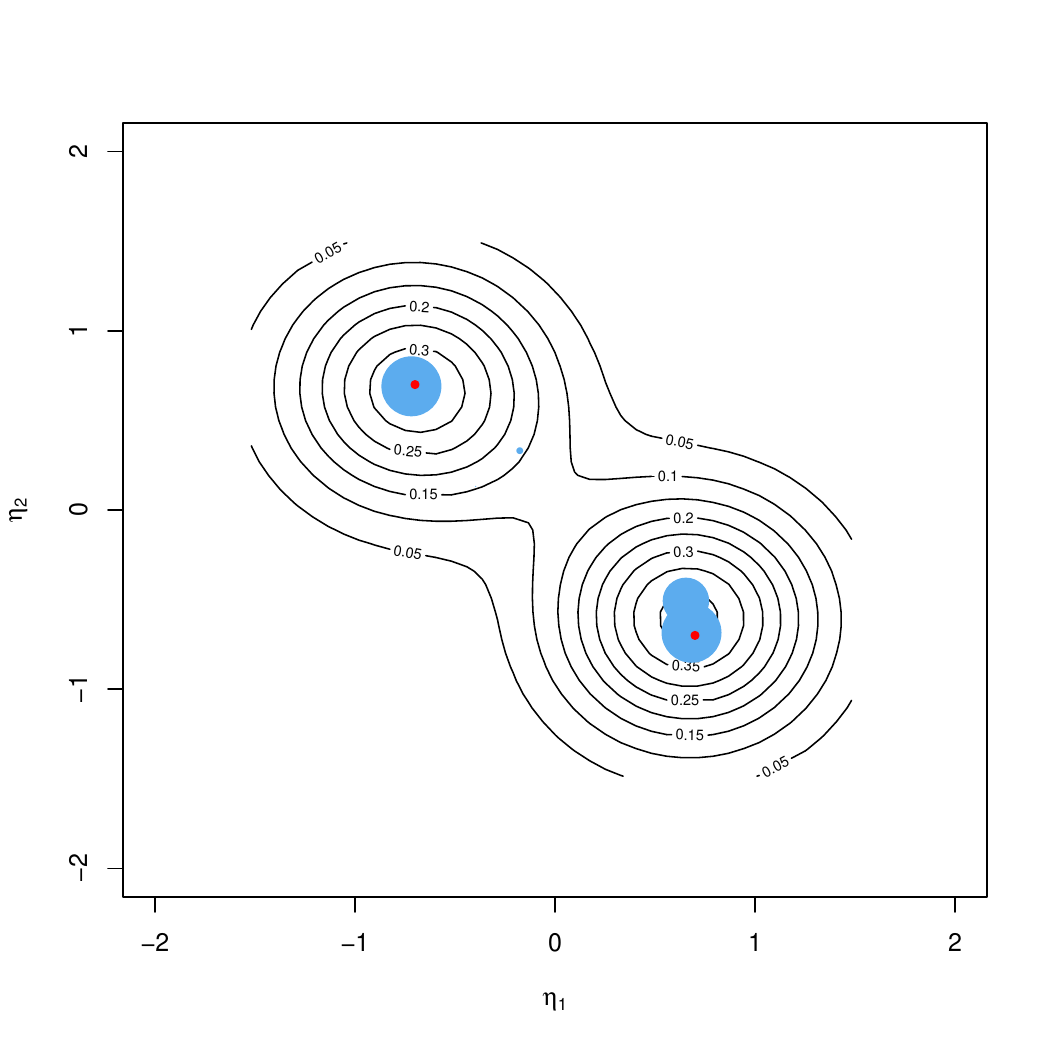}
	\caption{Mass points and smoothed contours of the NPMLE for the discrete Gautier-Kitamura
	simulation setting with mass concentrated at $(0.7,-0.7)$ and $(-0.7,0.7)$ as indicated by the red circles.
	The mass points of the NPMLE are indicated by the solid blue circles, and contours of a smoothed
	version of the NPMLE using a bivariate Gaussian kernel.}
	\label{fig.KW0}
\end{figure}
\begin{figure}
	\includegraphics[scale=0.8]{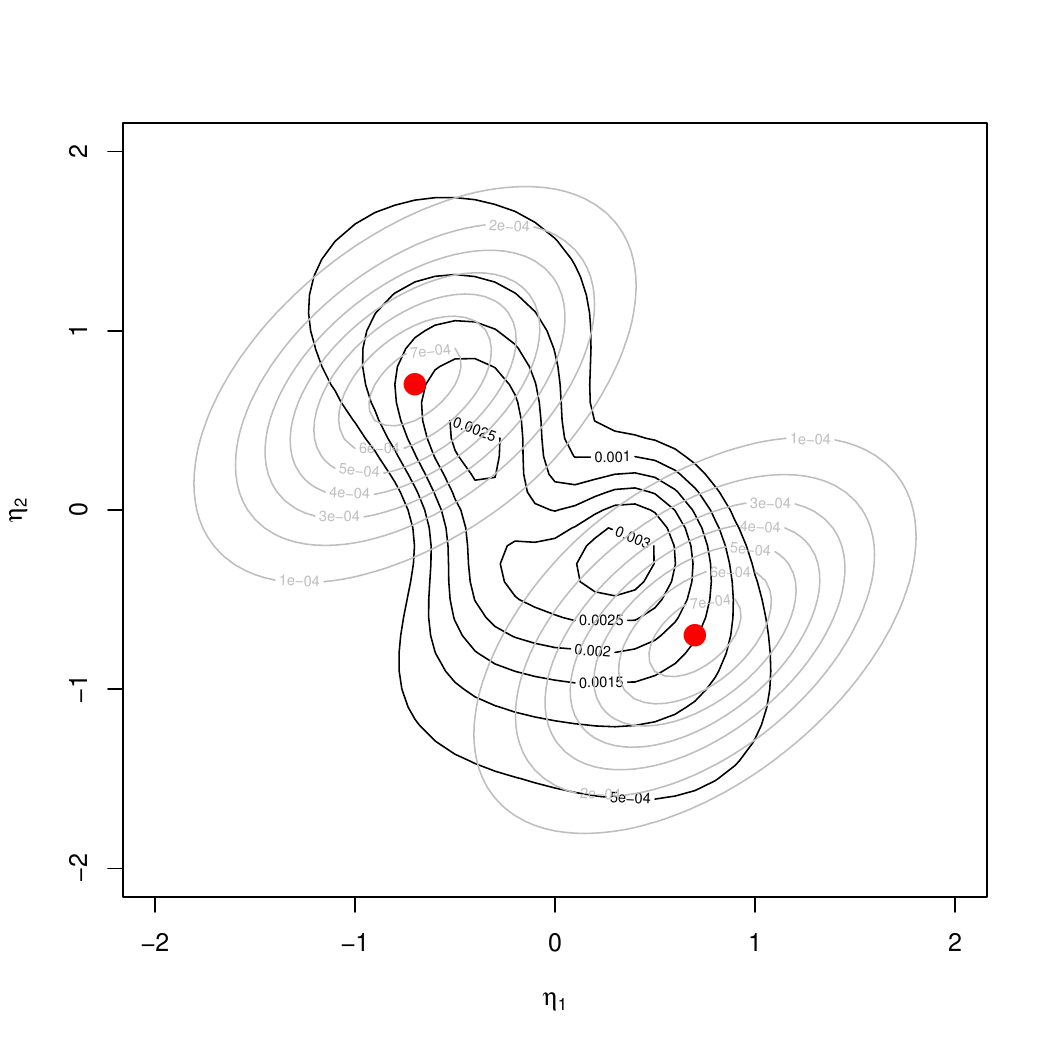}
	\caption{Gautier-Kitamura contours for a sample from the smooth bimodal bivariate distribution:  
	The true distribution of the random coefficients is a Gaussian location mixture 
	with two components each with variance 0.3, covariance 0.15 and means
	$(0.7,-0.7)$ and $(-0.7,0.7)$.  Contours of the true density are indicated 
	in grey with respective means by the solid red circles.} 
	\label{fig.GK1}
\end{figure}
\begin{figure}
	\includegraphics[scale=0.8]{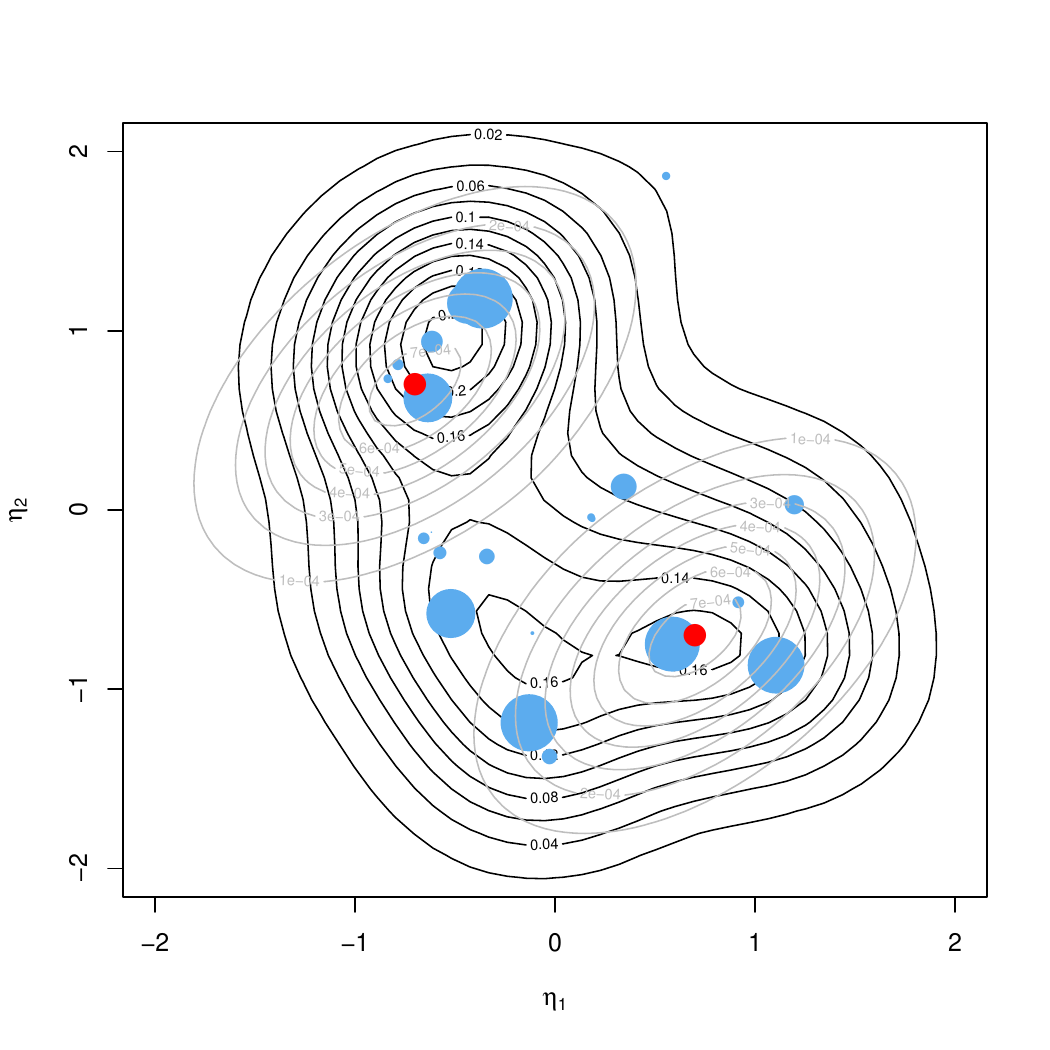}
	\caption{The NPMLE for a sample from a smooth bivariate distribution:  The
	true distribution of the random coefficients is a Gaussian location mixture 
	with two components each with variance 0.3, covariance 0.15 and means
	$(0.7,-0.7)$ and $(-0.7,0.7)$.  Contours of the true density are indicated 
	in grey with respective means by the solid red circles.   The mass points
	of the unsmoothed NPMLE are indicated by the solid blue circles whose
	areas depict associated mass.  Contours of the smoothed NPMLE are shown in black.}
	\label{fig.KW1}
\end{figure}

A somewhat more challenging setting for the NPMLE, taken directly from \citeasnoun{GK}, involves random
coefficients that are generated from a mixture of two bivariate Gaussians with the same centers
as in the previous case, but now both with variances 0.3, and covariance 0.15 for each of the equally
weighted components.
In Figure \ref{fig.GK1} we depict the density contours of the Gautier-Kitamura estimator using their
recommended tuning parameters, $T = 3$ and $TX = 10$,
with the contours of the true density of the random coefficients shown in grey.  Again, it is
apparent that the truncated basis expansion tends to shrink the mass of the estimated distribution
toward the origin.
In Figure \ref{fig.KW1}  we illustrate estimated mass points as well as contours of the smoothed NPMLE 
density again after convolution with a bivariate Gaussian distribution with diagonal covariance
matrix with entries $(0.04, 0.04)$.  The contours of the true density of the random coefficients 
are depicted by the grey contours with centers indicated by 
the two red circles.  The unsmoothed NPMLE has discrete mass points indicated by
the blue circular regions in the figure.  The smoothing introduces a tuning parameter into
the NPMLE fit, but it should be stressed that prior to the convolution step to impose the
smoothing there is no tuning parameter selection required.  
Clearly, there is more dispersion as we might expect in the NPMLE discrete solution, but the smoothed estimate
quite accurately captures the two modes of the random coefficient density.
Careful examination of the contour labeling of Figures \ref{fig.GK1} and \ref{fig.KW1} reveals
that the Gautier-Kitamura contours are concentrated along the axis connecting the two Gaussian centers
and assign almost no probability to the regions with $\eta_1 < -1$ or $\eta_1 > 1$, in contrast the NPMLE
contours cover the effective support of the true distribution somewhat better.

One replication of an experiment doesn't offer sufficient evidence of performance, 
so we have carried out two small simulation experiments to compare
performance of the various estimators under consideration using both of the foregoing simulation 
settings.  In the first experiment data is generated in accordance with the
two point discrete distribution with sample size $n = 500$.  
Four estimators are
considered:  two variants of the NPMLE, one smoothed the other not, referred to in the tables as
NPMLE and NPMLEs respectively, the Gautier-Kitamura (GK) estimator
with default tuning parameter selection, the classical logistic regression estimator and the Klein-Spady (KS)
single index estimator as implemented in the R package of \citeasnoun{np}.  For each
estimator we compute predicted probabilities for a fresh sample  of 500 $x$ observations also drawn from
the same Gaussian distribution generating the data used for estimation, and these probabilities are
compared with the true probabilties of $P(Y_i = 1 | X = x_i)$ for this sample.  Table \ref{tab.simKS0a} reports
mean absolute and root mean squared errors for the predicted probabilities over the 100 replications of the
experiment.  The discrete NPMLE is the clear winner, with its smoothed version performing only slightly
better than the Gautier-Kitamura deconvolution procedure. 

Skeptical minds may, correctly, regard the two point distribution simulation setting as ``too favorable''
to the NPMLE since the NPMLE is known to deliver a relatively sparse discrete estimate.
Thus, it is of interest to see how our comparison would look when the true random coefficient
distribution is itself smooth.  

%latex.default(Tab, file = "simKS0a.tex", rowlabel = "", where = "!htbp",     label = "tab.simKS0a", caption = cap, caption.loc = "bottom")%
\begin{table}[!htbp]
\begin{center}
\begin{tabular}{lrrrrr}
\hline\hline
\multicolumn{1}{l}{}&\multicolumn{1}{c}{GK}&\multicolumn{1}{c}{NPMLE}&\multicolumn{1}{c}{NPMLEs}&\multicolumn{1}{c}{Logit}&\multicolumn{1}{c}{KS}\tabularnewline
\hline
MAE&$0.1211$&$0.0347$&$0.1064$&$0.1684$&$0.2698$\tabularnewline
RMSE&$0.1532$&$0.0796$&$0.1428$&$0.2042$&$0.3450$\tabularnewline
\hline
\end{tabular}
\caption{Bivariate Point Mass Simulation Setting:  
Mean Absolute and Root Mean Squared Errors of Predicted Probabilities\label{tab.simKS0a}}\end{center}
\end{table}

%latex.default(D, file = "simKS1a.tex", rowlabel = "", where = "!htbp",     label = "tab.simKS1a", caption = cap, caption.loc = "bottom")%
\begin{table}[!htbp]
\begin{center}
\begin{tabular}{lrrrrr}
\hline\hline
\multicolumn{1}{l}{}&\multicolumn{1}{c}{GK}&\multicolumn{1}{c}{NPMLE}&\multicolumn{1}{c}{NPMLEs}&\multicolumn{1}{c}{Logit}&\multicolumn{1}{c}{KS}\tabularnewline
\hline
MAE&$0.1288$&$0.0592$&$0.0475$&$0.0709$&$0.2108$\tabularnewline
RMSE&$0.1440$&$0.0748$&$0.0594$&$0.0896$&$0.2626$\tabularnewline
\hline
\end{tabular}
\caption{Bivariate Gaussian Simulation Setting:  
Mean Absolute and Root Mean Squared Errors of Predicted Probabilities\label{tab.simKS1a}}\end{center}
\end{table}

Table \ref{tab.simKS1a} reports mean absolute and root mean squared errors for the predicted probabilities 
for the 100 replications of the second simulation setting with the location mixture of Gaussians.
Again, the NPMLE is the clear winner, but now its smoothed version performs 
somewhat better than the unsmoothed version although both do better than the Gautier-Kitamura 
deconvolution procedure.

\section{An Application to Modal Choice}

In this section we revisit a modal choice model of \citeasnoun{Horowitz93}, motivated by similar
considerations as the classical work of \citeasnoun{McFadden} on the Bay Area Rapid Transit system.
The data consists of 842 randomly sampled observations of individuals' transportation choices for 
their daily journey to work in the Washington DC metro area. 
Following Horowitz, we focus on the binary choice of commuting to work by automobile 
versus public transit. 
In addition to the individual mode choice variable, $y_i$, taking the value 1 if an automobile is used 
for the journey to work and 0 if public transit is taken, we observe the number of cars owned by the 
traveller's household (\textit{AUTOS}), 
the difference in out-of-vehicle (\textit{DOVTT}) and in-vehicle travel time (\textit{DIVTT}).
Differences are expressed as public transit time minus automobile time in minutes per trip. 
The corresponding differences in transportation cost, \textit{DCOST}, public transit fare minus automobile 
travel cost is measured in cents per trip. 
We have omitted the variable $\textit{DIVTT}$ from our analysis since it had no significant impact
on modal choice in prior work, see Table 2 of \citeasnoun{Horowitz93} for estimation results using 
various parametric and semiparametric models. Although our methodology can accommodate  additional 
$x_i$ variables with random coefficients it becomes considerably more difficult to visualize 
distributions of random coefficients $F_\eta$ in higher dimensions.  Other control variables
in the vector $w_i$ could also be accommodated, but the application doesn't offer obvious candidates.

We consider the following random coefficient binary choice model: 
\[
\P(y_i=1 \mid x_i, v_i, \textit{AUTOS}_i=k ) = \int 1\{ x_i ^\top \eta_i - v_i \geq 0\}dF_{\eta, k}
\]
where $x_i = (1, \textit{DOVTT}_i)$ and $v_i = \textit{DCOST}_i/100$ and $\eta_i = \{\eta_{1i}, \eta_{2i}\}$. 
We have normalized the coefficient of $v_i$ to be 1, since $v_i$ represents a negative price, 
transit fare minus automobile cost. Under this normalization, the coefficient $\eta_{2i}$ 
has a direct interpretation as the commuter's value of travel time in dollars/minute. 
The coefficient $\eta_1$ obviously has the same units as $v_i$, and can be interpreted as
a threshold -- setting a critical value for $v_i$ above which the subject decides to commute
by automobile, and below which he chooses to take public transit, assuming that the time differential
is negligble.  Auto ownership is a discrete variable, taking values between 0 and 7. 
Households with 3 cars or more commute exclusively by automobile, so we only consider subjects with 
fewer than 3 cars, a subsample containing about 90\% of the data. 
Since car ownership is plausibly an endogenous decision and may act as a proxy for wealth of the 
household and potential constraint on the travellers' mode choices, 
we estimate distinct distributions of the random coefficients for subjects with zero, one and two cars. 
Figure \ref{fig:data} provides scatter plots of \textit{DOVTT} and \textit{DCOST} for $k \in \{0,1,2\}$,
distinguishing auto and transit commuters by open and filled circles. 

\begin{figure}[h]
	\includegraphics[width=\textwidth]{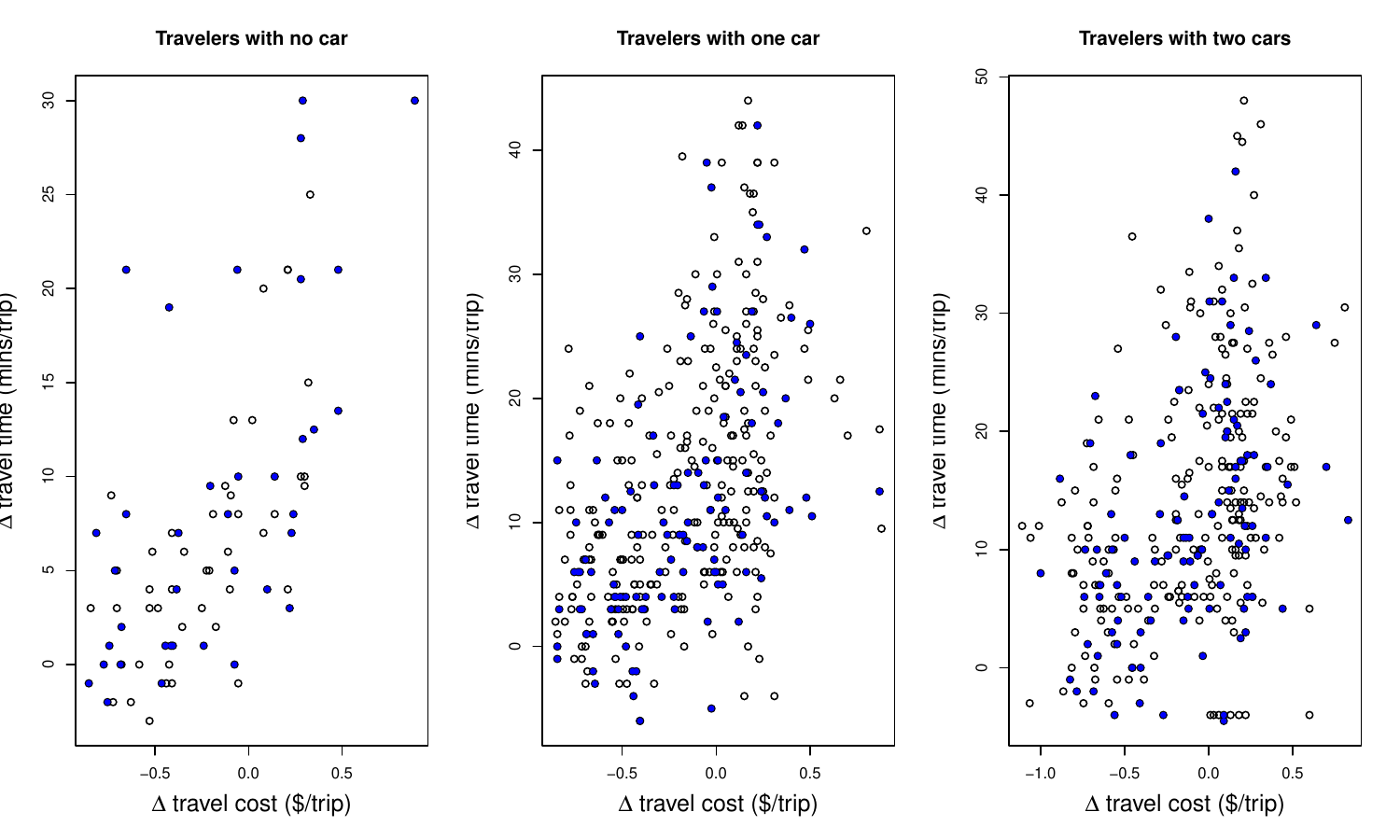}
	\caption{Scatter plot of DOVTT and DCOST for commuters with different number of cars at home: 
	the open circles correspond to subjects that commute by auto while the solid (blue) points 
	represent those who use public transit.}
	\label{fig:data}
\end{figure}

We briefly report results for the subsamples, $k \in \{0,1,2\}$, separately, focusing initially on the
shape and dispersion of the estimated $F_\eta$ distributions.  F or each subsample we contrast the 
discrete distribution delivered by the NPMLE with the contours of the smooth density produced by the
Gautier and Kitamura estimator.  

\subsection{Commuters without an automobile}

There are 79 observations for commuters without a car. Despite having no car 16 still manage 
to commute to work by automobile. The lines determined by these realizations of $(x_i, v_i)$ 
lead to a partition of $\mathbb{R}^2$ into 2992 polygons, of these only 112 are locally maximal 
and therefore act as potential candidates for positive mass assigned by the NPMLE of $F_{\eta}$. 
In Figure \ref{fig:car0} we compare the estimates of $F_\eta$ produced by the NPMLE and the
deconvolution estimator of Gautier and Kitamura.  The solid (red) points in the figure represent
the locations of the mass points identified by the NPMLE;  the mass associated  with each of these
points is reported in Table \ref{tab:car012}.  Only 11 of the 112 candidate polygons achieve mass greater
that 0.001, determined by the NPMLE.  Although we plot points with area proportional to estimated mass,
we should again recall that these mass points are really associated with underlying polygons.

In contrast,
the Gautier-Kitamura density contours are entirely concentrated near the origin.  We have experimented
quite extensively with the choice of tuning parameters for the Gautier-Kitamura estimator, eventually
adopting a likelihood criterion for the choice of the sieve dimensions, $T$ and $TX$, that are required.  This
criterion selects rather parsimonious models in this application, choosing $T =2$ and $TX = 3$ for
this subsample.  See Table \ref{tab:GKTuning} of the Appendix for further details on this selection.
Selection of lower dimensional Fourier-Laplace expansions obviously yield more restrictive parametric
specifications, however this greater degree of regularization seems to be justified by the commensurate
reduction in variability of the estimator.  Although the comparison is inherently somewhat unfair we note 
that the NPMLE achieves a log-likelihood of -28.16, while the Gautier-Kitamura estimate achieves -37.58.

The two points on the far left of Figure \ref{fig:car0} constitute about 0.05 mass each and represent 
individuals who seem to be committed transit takers.  A coefficient of, say $\eta_1 = -8$ would mean
that the transit fare per trip would have to be 8 dollars per trip higher than the corresponding
car fare to induce them to travel by car.  The fact that the $\eta_2$ coordinates associated with these
extreme points is about one, means that, since the variable $DOVTT$ measures the transit time differential 
in its original scale of minutes, for such individuals a 10 minute time differential would be sufficient
to induce them to commute by automobile.

\begin{figure}[h]
	\includegraphics[scale=0.5]{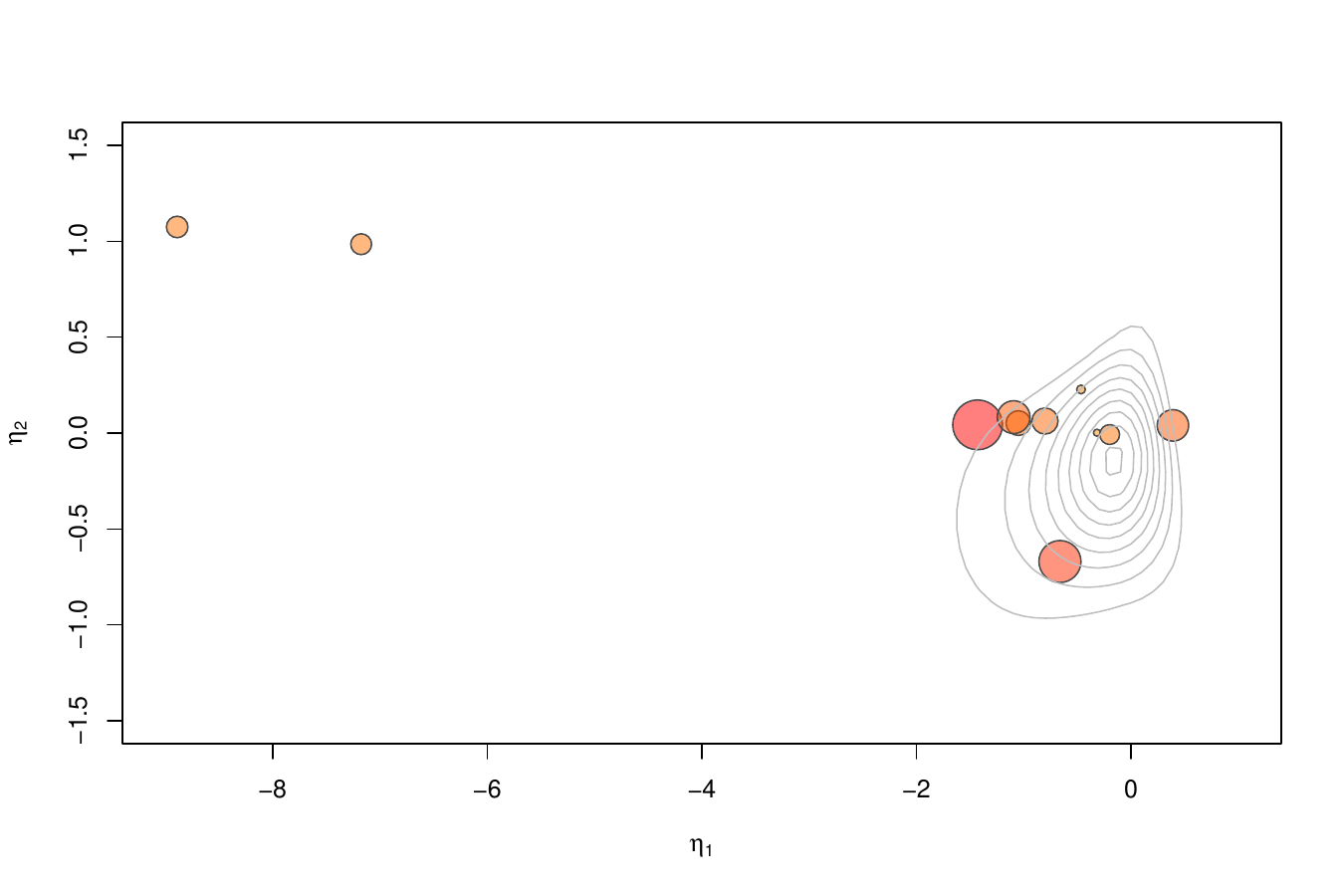}
	\caption{Two Estimates of the Random Coefficient Distribution $F_\eta$
	Based on the Subsample of Commuters with No Automobile:
	Shaded circles represent the interior points of polygons with positive mass
	as estimated by the NPMLE.  The area inside the circles is proportional to estimated mass.
	Table \ref{tab:car012} in the Appendix reports the NPMLE results in greater detail.
	Grey contour lines depict the density contours of the Gautier-Kitamura estimator
	using Fourier-Laplace tuning parameters $T = 2$ and $TX = 3$ selected by the
	log-likelihood criterion.}
	\label{fig:car0}
\end{figure}

\subsection{Commuters with One Automobile}

There are 355 commuters who have one automobile of which 302 commute by car.
The hyperplane arrangement  determined by this subsample of pairs $(x_i, v_i)$ yields a 
tessellation of $\mathbb{R}^2$ into 55549 distinct polygons 
of which there are 1272 with locally maximal counts. 
Figure \ref{fig:car1} displays the estimated mass points of the NPMLE and the contour plot 
the Gautier-Kitamura density estimate as in the preceeding figure.  
As for the subsample without an automobile, the NPMLE mass is considerably more dispersed than
the Gautier-Kitamura contours.  This may be partly attributed to the rather restrictive choice
of the tuning parameters, $T = 2$ and $TX = 3$, dictated by the likelihood criterion.
Again, a more detailed tabulation of how the NPMLE mass is allocated is available in Table
\ref{tab:car012}.  It may suffice here to note that while most of the NPMLE mass is again
centered near the origin, there is about 0.10 mass at $(\eta_1, \eta_2) \approx (9.7, -0.24)$
and another, roughly, 0.05 probability with $\eta_1 < -12$.  The Gautier-Kitamura contours are again
much more concentrated around the origin.  These differences are reflected in substantial differences
in predicted outcomes and estimated marginal effects. 

\begin{figure}[h]
	\includegraphics[scale=0.5]{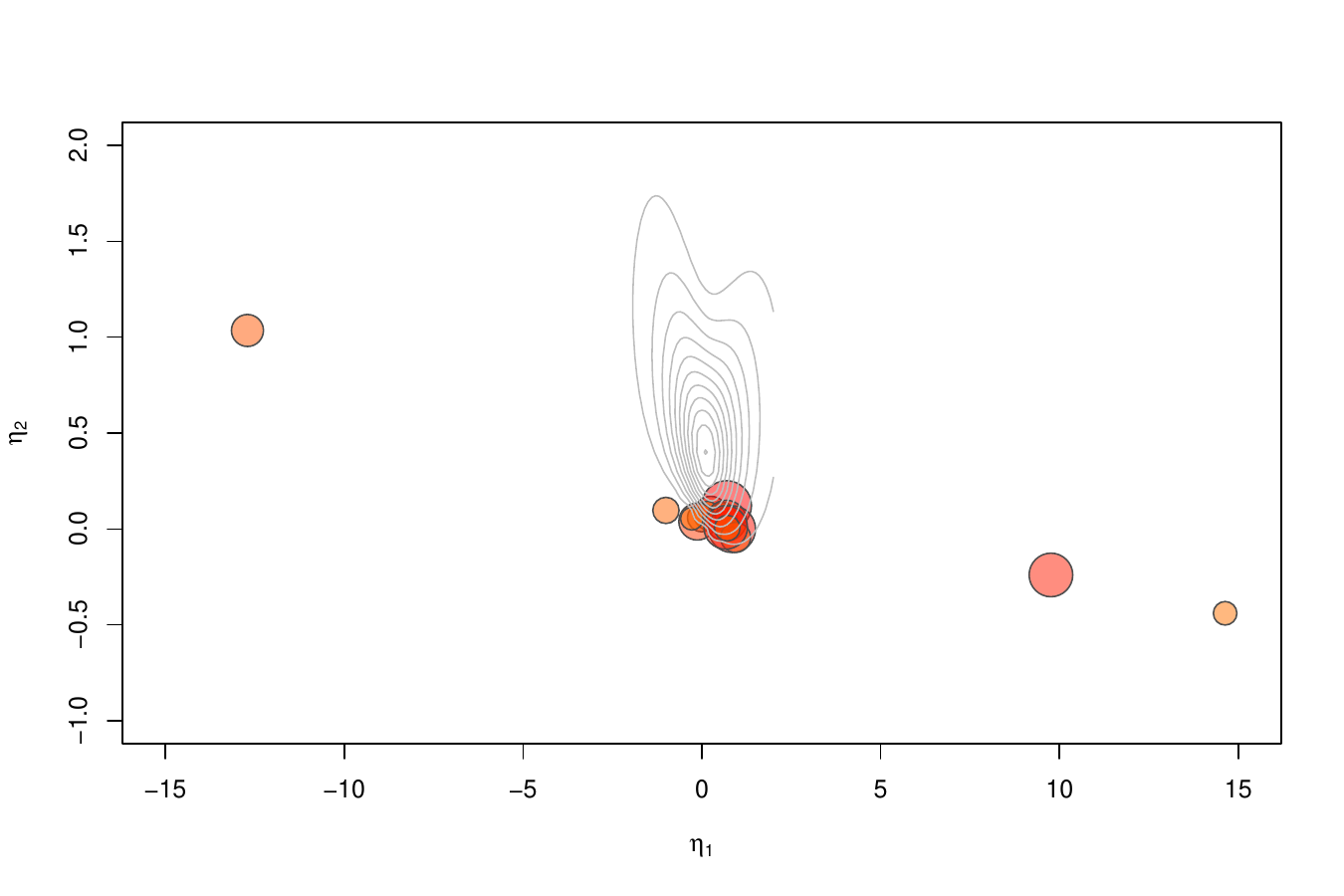}
	\caption{Two Estimates of the Random Coefficient Distribution $F_\eta$
	Based on the Subsample of Commuters with One Automobile:
	Shaded circles represent the interior points of polygons with positive mass
	as estimated by the NPMLE.  The area inside the circles is proportional to estimated mass.
	Table \ref{tab:car012} in the Appendix reports the NPMLE results in greater detail.
	Grey contour lines depict the density contours of the Gautier-Kitamura estimator
	using Fourier-Laplace tuning parameters $T = 3$ and $TX = 3$ selected by the
	log-likelihood criterion.}
	\label{fig:car1}
\end{figure}

\subsection{Commuters with Two Automobiles}

There are 316 travellers with 2 cars of which 303 commute to work by automobile. 
Of the 44662 polygons for this subsample there are only 288 with locally maximal counts. 
Figure \ref{fig:car2} depicts the  mass points of the NPMLE and the contours of
the Gautier-Kitamura estimate for this subsample.  The disperson of the threshhold
parameter $\eta_1$ is considerably smaller than for the zero and one car subsamples, but it
is still the case that the NPMLE is more dispersed that the Gautier-Kitamura estimate in this dimension.
Curiously, the Gautier-Kitamura estimate places most of its mass well above any of the NPMLE mass points.
This may again be a consequence of the low dimensionality of the Fourier-Laplace expansion, which is 
selected as $T = 2$ and $TX = 3$ by the likelihood criterion.

\begin{figure}
	\centering
	\includegraphics[scale=0.5]{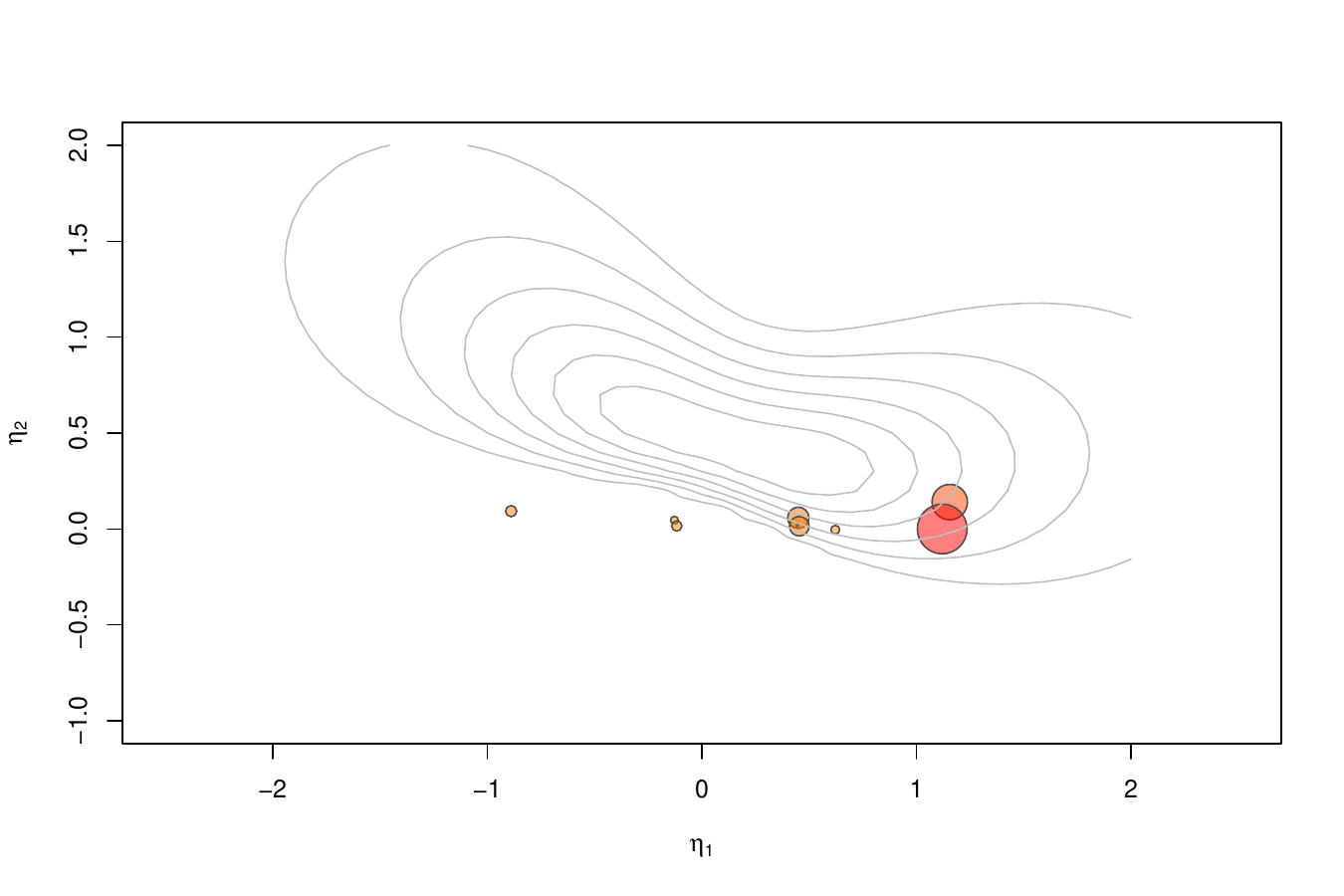}
	\caption{Two Estimates of the Random Coefficient Distribution $F_\eta$
	Based on the Subsample of Commuters with Two Automobiles:
	Shaded circles represent the interior points of polygons with positive mass
	as estimated by the NPMLE.  The area inside the circles is proportional to estimated mass.
	Table \ref{tab:car012} in the Appendix reports the NPMLE results in greater detail.
	Grey contour lines depict the density contours of the Gautier-Kitamura estimator
	using Fourier-Laplace tuning parameters $T = 2$ and $TX = 3$ selected by the
	log-likelihood criterion.}
	\label{fig:car2}
\end{figure}

\subsection{Marginal Effects}

With discrete choice model, a common parameter of interest is the marginal effect of some control variables. 
We consider two scenarios for evaluating marginal effects based on estimates of the quantities,
\begin{align*}
	\Delta_z(z_0, v_0,\Delta z) = \PP (y = 1\mid v = v_0, z=z_0) - \PP(y=1 \mid v = v_0, z = z_0 - \Delta z)\\
	\Delta_v(z_0,v_0, \Delta v) = \PP(y = 1\mid v = v_0, z = z_0) - \PP(y=1\mid v = v_0 - \Delta v, z = z_0).
\end{align*}
The value $\Delta_z(z_0,v_0, \Delta z)$ measures the marginal effect of reducing out-of-vehicle travel time 
by $\Delta z$ minutes/trip; the value $\Delta_v(z_0, v_0, \Delta v)$ measures the marginal effect of reducing the 
transit fare by $\Delta v$ dollars holding transportation time constant. In each case, we fix the initial values 
$(z_0,v_0)$ at the 75-th quantiles for the subsample of individuals who drive to work.  
Figures \ref{fig:MEv} and \ref{fig:MEt} depict the marginal effect of fare reduction and commute time reduction, 
respectively, conditional on automobile ownership. 

\begin{figure}
	\centering
	\includegraphics[scale=0.6]{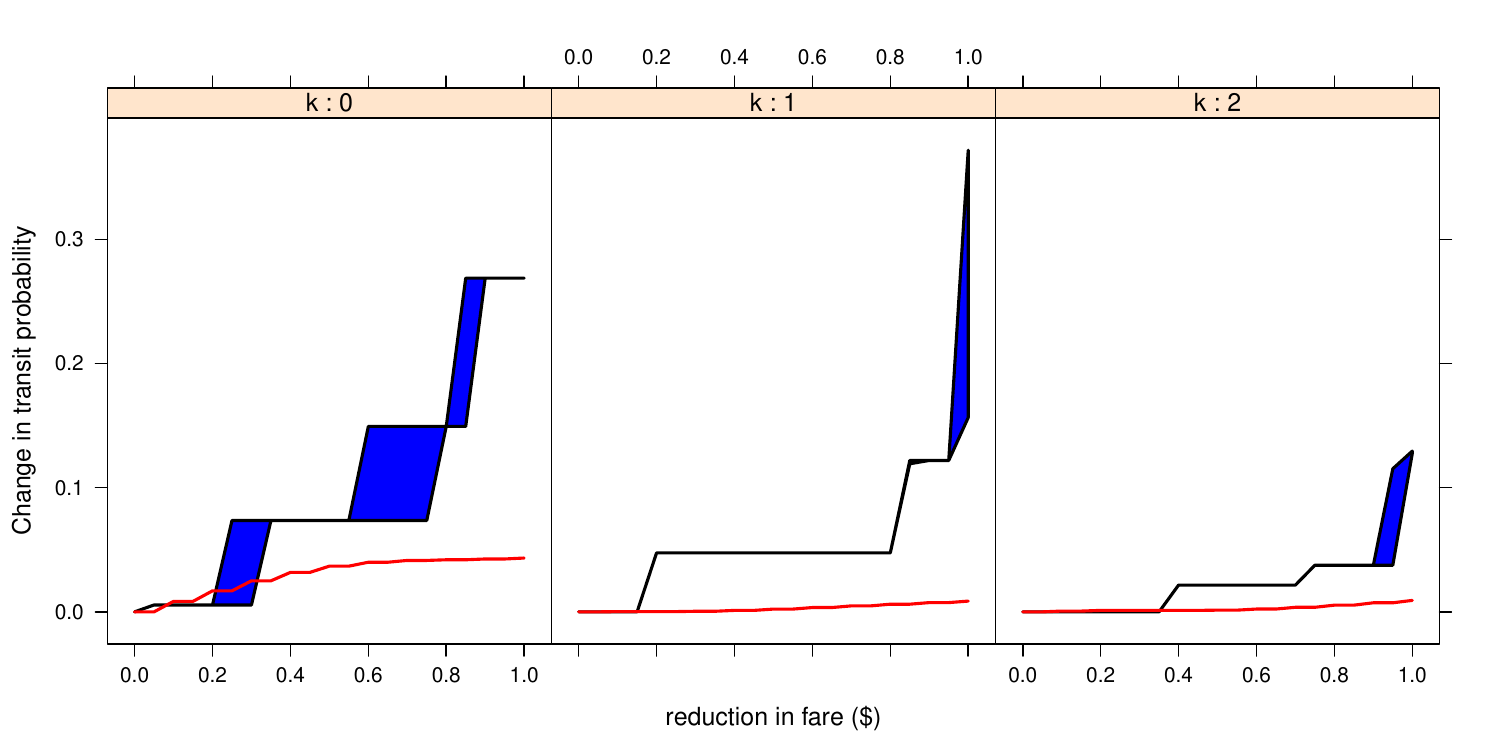}
	\caption{Set Valued Estimates of Marginal Effect for Transit Fare Reduction: 
	The shaded solid regions represent the NPMLE set-valued estimates of the marginal
	effect of reducing the transit fare on the probability of choosing the transit
	option, the (red) line represents the corresponding estimates from the Gautier-Kitamura fit.}
	\label{fig:MEv}
\end{figure}

\begin{figure}
	\centering
	\includegraphics[scale=0.6]{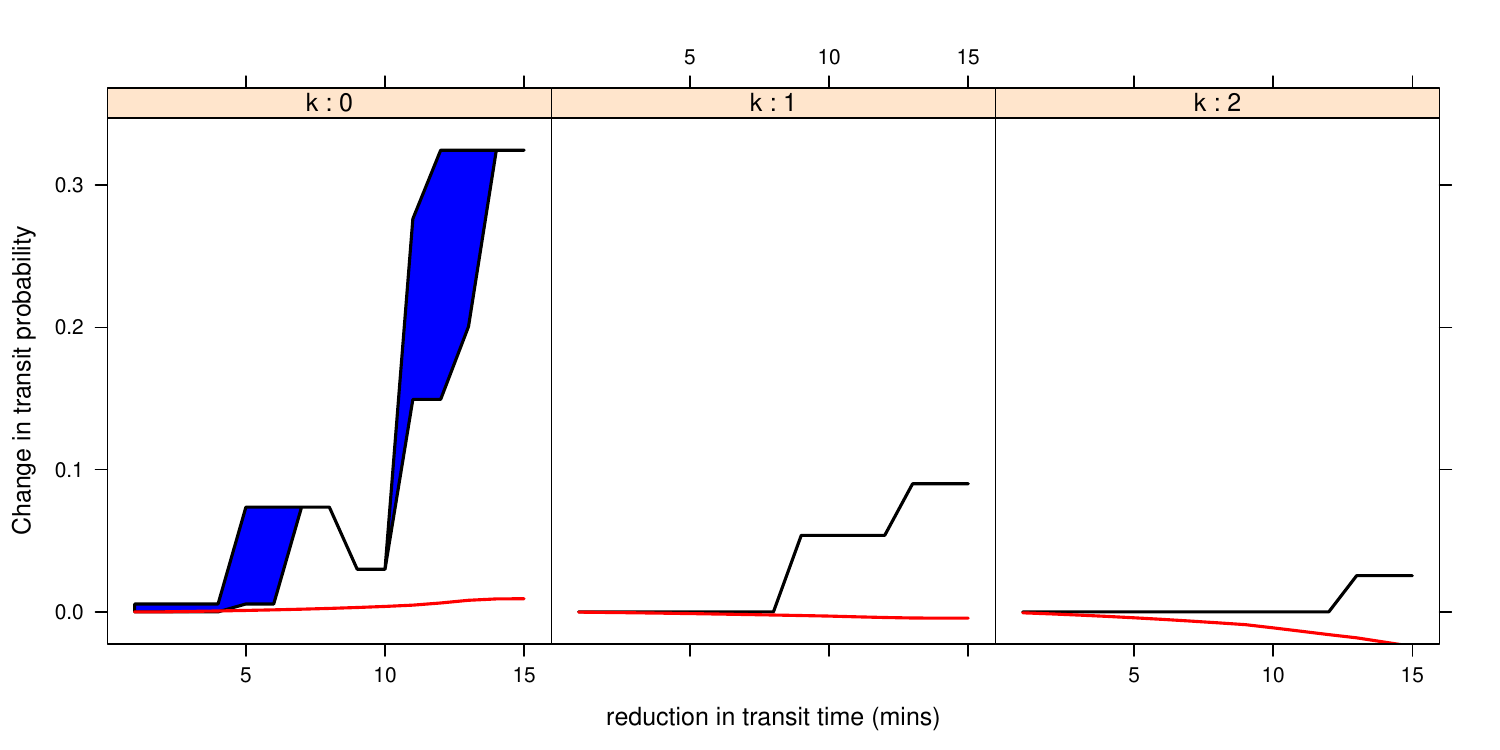}
	\caption{Set Valued Estimates of Marginal Effect for Transit Time Reduction: 
	The shaded solid regions represent the NPMLE set-valued estimates of the marginal
	effect of reducing the transit time in minutes on the probability of choosing the transit
	option, the (red) line represents the corresponding estimates from the Gautier-Kitamura fit.}
	\label{fig:MEt}
\end{figure}

As discussed in Section \ref{sec.asymptopia}, for any fixed $n$, the NPMLE $\hat F_n$ assigns
probability mass $\{\hat p_j\}_{j = 1:M}$ to polytopes $\{C_j\}_{j=1:M}$ that define
the partition of the parameter space determined by the hyperplane arrangement, not to specific points. 
This feature of the NPMLE naturally leads to a set valued estimator for marginal effects for any finite sample. 
Suppose we would like to estimate $\Delta_z(z_0,v_0, \Delta z)$.  Denoting the half-space determined by any point 
$(1, z, v)$ by $H^+(v,z) := \{(\eta_1, \eta_2): \eta_1 + \eta_2 z - v \geq  0 \}$,
it is easy to see that the set estimator for $\mathbb{P}(y=1\mid v=v_0, z = z_0)$ can be expressed as,
\[
\mathbb{P}_{\hat F_n}(H^+(v_0,z_0))  \in  [\hat L_n(v_0,z_0), \hat U_n(v_0,z_0)]
\]
with \[
\hat L_n = \sum_{j=1}^M 1\{C_j \subseteq H^+(v_0,z_0)) \hat p_j
\]
and 
\[
\hat U_n =  \sum_{j=1}^M 1\{C_j \subseteq H^+(v_0,z_0)) \hat p_j +
\sum_{j=1}^M 1\{C_j \nsubseteq H^+(v_0,z_0), C_j \cap H^+(v_0,z_0)\neq \emptyset\}\hat p_j
\]
These bounds are constructed by finding the corresponding polygons crossed by the new hyperplane $H(v_0, z_0)$. The lower bound sums all non-zero probability masses $\hat p_j$ allocated to the polygons that are completely contained in the half-space $H^{+}(v_0,z_0)$ while the upper bound adds in the additional non-zero probably masses that are allocated to polygons that are crossed by $H(v_0, z_0)$.
	
	The set valued estimator for the marginal effect is therefore, 
\[
\hat \Delta_z(z_0,v_0, \Delta z) \in [\hat  L_n(z_0,v_0) - 
\hat U_n(z_0- \Delta z,v_0), \hat U_n(z_0,v_0) - \hat L_n(z_0 - \Delta z, v_0)]
\]
Figures \ref{fig:MEv} and \ref{fig:MEt}  report these bounds for different values of $\Delta z$ and $\Delta v$. 

The corresponding marginal effects for the Gautier-Kitamura estimates are depicted as the dotted red
curves in this figure.  The concentration of the Gautier-Kitamura 
$\hat F_\eta$ near the origin implies marginal effects that are considerably smaller than those implied
by the NPMLE results.  For both the fare and transit time effects
the Gautier-Kitamura estimates.  Car ownership is clearly an important influence especially on
the marginal effects of time savings for commuters without a car; while there is essentially no
marginal effect for commuters with two cars.

\subsection{Single Index Model}
As a final comparison, we reconsider the single index model described in Section 2 
where the parameter $\eta_2$ is treated as fixed and there is only a random intercept effect,
\[
\mathbb{P}(y_i = 1\mid x_i, v_i, AUTOS_i = k) = \int 1\{\eta_{1i} + z_i \eta_2 - v_i \geq 0 \} dF_{\eta_1,k}.
\]
We consider several semiparametric estimators that make no distributional assumption on $F_{\eta_1,k}$ 
as well as the parametric probit estimator that presumes that $F_{\eta_1,k}$ is standard Gaussian. 
Since we can only identify $F_{\eta_{1,k}}$ up to scale, we again normalize the coefficient for $v_i$ to be 1. 
We consider the kernel-smoothing based estimator proposed by \citeasnoun{KleinSpady} 
and the score estimator proposed in \citeasnoun{GH} based on the nonparametric maximum likelihood 
estimator of $F_{\eta_1,k}$ as in \citeasnoun{Cosslett83}.  The former estimator requires a choice of
a bandwidth for the kernel estimate, whereas the latter is free of tuning parameters.
%We provide in Figure \ref{fig: Cprofile} a visualization of the score function for $k \in \{0,1,2\}$ 
%on a equally spaced grid on $[-0.6, 0.4]$ with width 0.001 for $\eta_2$. Both estimators are shown to 
%be $\sqrt{n}$-consistent and asymptotically normal. 
The Klein-Spady estimates are based on the implementation in the R package  {\tt{np}} of 
\citeasnoun{np} using a Gaussian kernel.  The bandwidth is chosen automatically via a
likelihood-based cross-validation criterion. The estimation results 
are reported in Table \ref{tab:singleindex} together with the probit model estimates. In comparison, 
we also include in the last column the log-likelihood of the random coefficient model where $\eta_2$ 
is allowed to be heterogeneous across individuals. 

A virtue of the single index representation is that it is possible to estimate standard errors for
the fixed parameter  estimate $\hat \eta_2$, which appear in the table in parentheses under their
coefficients.  However, as is clear from the foregoing figures and
the log likelihood values of Table \ref{tab:singleindex} these standard errors require a willing
suspension of disbelief in view of the apparent heterogeneity of the NPMLE estimates
of the bivariate model.

% Code for this table:  /Users/roger/projects/rcbr/new/rcbr/application_nojitter2/Horowitz93_singleindex2.R
\begin{table}[ht]
\begin{tabular}{r|rr|rr|rr|c}
\hline
Cars& \multicolumn{2}{c}{Groeneboom-Hendrickx} & \multicolumn{2}{c}{Klein-Spady}& \multicolumn{2}{c}{Probit}& \multicolumn{1}{c}{NPMLE$(\eta_1, \eta_2)$}\\
\hline
& $\hat \eta_2$ & logL & $\hat \eta_2$ & logL & $\hat \eta_2$& logL & logL\\ 
\hline
0&$\underset{(0.022)}{0.026}$& -32.87 &$\underset{(0.026)}{-0.396}$& -37.60 &$\underset{(0.021)}{ 0.034}$ & -37.420 & -29.55\\ 
1&$\underset{(0.006)}{0.018}$& -121.71&$\underset{(0.007)}{0.034}$& -131.71 &$\underset{(0.010)}{ 0.028}$ & -130.84 & -112.32\\ 
2&$\underset{(0.009)}{0.030}$& -47.33&$\underset{(0.003)}{0.003}$& -51.85 &$\underset{(0.019)}{ 0.048}$ & -51.80 & -46.13\\ 
\hline
\end{tabular}
\vspace{5mm}
\caption{Estimates for $\eta_2$ of the single index model for households having different numbers of vehicles. 
%For the semiparametric single index model, we use the Groeneboom-Hendrickx and the Klein and Spady estimator. 
The semiparametric and parametric probit estimates normalize the coefficient of $v$  to be 1. 
The last column reports the log-likelihood of the NPMLE  for the bivariate model in which $\eta_2$ is allowed 
to be individual specific. The Klein-Spady estimator is implemented with the {\tt{npindexbw}} and 
{\tt{npindex}} functions of the {\tt{np}} package. We use a Gaussian kernel and the bandwidth is chosen 
based on optimizing the likelihood criteria for both parameters and the bandwidth through leave-one-out 
cross validation. The BFGS method was used for optimization with 20 randomly chosen starting initial values.
The Groeneboom-Hendrickx results were computed with the {\tt{GH}} function from the R package {\tt{RCBR}}.
%\citeasnoun{RCBR}.
}
\label{tab:singleindex}
\end{table}

\section{Conclusion}
Random coefficient binary response models estimated by the nonparametric maximum likelihood methods
of \citeasnoun{KW} as originally proposed by \citeasnoun{Cosslett83} and extended by \citeasnoun{Ichimura} 
offer a flexible alternative to established parametric binary response methods potentially revealing
new sources of preference heterogeneity.  Modern convex optimization 
methods combined with recent advances in the algebraic geometry of hyperplane arrangements provide efficient 
computational techniques for the implementation of these methods.  Extensions of these methods to 
multinomial response, and causal models with endogenous covariates offer challenging opportunities
for further investigation. 

\section{Acknowledgements}
The authors would like to express their appreciation to Frederico Ardila for his guidance toward the
relevant combinatorial geometry literature, and to Steve Cosslett, Hide Ichimura and Yuichi Kitamura
for their pioneering work on the random coefficient binary response model. 
Thomas Stringham provided very capable research assistance.  We also would like to
thank Ismael Mourifie, Yuanyuan Wan and Stanislav Volgushev for useful discussions.

\bibliography{rcbr}

\newpage
\appendix
\section{Proofs}
{\small{
\begin{proof}[Proof of Theorem \ref{thm:N1}] 
Suppose the NPMLE assigned positive probability mass $p$ to such an $C_j$. By re-assigning the probability 
mass $p$ to the cell in the set $N_j$ with a larger count the likelihood could be strictly increased,
hence $p$ must be zero. 
\end{proof}

\begin{proof}[Proof of Theorem \ref{thm: zone}]
		Theorem 2.1 in \citeasnoun{Edelsbrunner93} with $k = 0$, yields,
		\[
		h \leq \binom{d-1}{0}\binom{n}{d-1} + \sum_{0 \leq j < d-1}\binom{j}{0}\binom{n}{j} = 
		\sum_{i=0}^{d-1} \binom{n}{i}.
		\]
\end{proof}

\begin{proof}[Proof of Theorem \ref{thm: identification}]
Given the model $y_i= 1\{w_i^\top \theta_0 +\beta_{1i} +  z_i^\top \beta_{-1i} \geq v_i\}$, 
we denote the random variable $\tilde U = W ^\top \theta_0 + \beta_{1i} + Z^\top \beta_{-1i}$. 
Under Assumptions \ref{Assump1} and \ref{Assump3} we can identify the conditional distribution of $\tilde U$ given 
$(W, Z) = (w, z)$ for all values of $(w, z)$ on its support. Fix $Z$ at some value $z$ and take values 
$w_1 \neq w_2$, we thereby identify $\theta_0$. To identify the distribution of $\beta_i$, 
consider that the characteristic function of $\tilde U| (W, Z) $ for all $t \in \mathbb{R}$ 
and any values of $(w,z)$ on its support is 
\[
	\phi_{\tilde U|(W,Z)}(t|w, z) = \mathbb{E}(e^{it\tilde U}| (W, Z) = (w, z)) 
	 = e^{itw^\top \theta_0} \mathbb{E}(e^{i(t\tilde z)^\top \beta}) 
	 = e^{itw^\top \theta_0} \phi_{\beta}(t\tilde z)
\]
where $\tilde z = \{1, z\}^\top$ and second equality holds under Assumption \ref{Assump1}. 
Under Assumption \ref{Assump3} and the fact that $\theta_0$ is already identified, 
the characteristic function $\phi_\beta$ is revealed by varying $\tilde z$ and 
hence the distribution of $\beta_i$ is identified.  This is essentially similar to the classical argument for
the Cram\'er-Wold device.
\end{proof}

\begin{proof}[Proof of Lemma \ref{lemma: continuity}]
For all $F \in \mathcal{F}_0$, since $w^\top \theta$ is continuous in $\theta$, then there exists an $N_1 > 0$ 
such that for $n \geq  N_1$, $| \mathbb{P}_F(H(x, w, \theta_n))-\mathbb{P}_F(H(x, w, \theta))| < \epsilon$. 
By the Portmanteau Theorem, for any $F \in \mathcal{F}_0$, we also have that $F_n \to F$ implies 
$\mathbb{P}_{F_n}(H(x, w, \theta)) \to \mathbb{P}_{F}(H(x, w,\theta))$. That means there exists 
$N_2 >0$ such that for $n \geq N_2$, $|\mathbb{P}_{F_n}(H(x, w, \theta)) - \mathbb{P}_{F}(H(x, w, \theta))| < \epsilon$. 
Therefore, $|\mathbb{P}_{F_n}(H(x, w, \theta_n))-\mathbb{P}_F(H(x, w, \theta))| < 2\epsilon$ for $n$ large enough. 
Since $\epsilon$ is arbitrary, it follows that 
\[
\underset{\gamma_n \to \gamma}{\lim} \mathbb{P}_{F_n}(H(x, w, \theta_n)) = \mathbb{P}_{F}(H(x, w, \theta))
\]
almost surely. 

\end{proof}

\begin{proof}[Proof of Lemma \ref{lemma: integrability}]
Since $0 \leq \mathbb{P}_{F}(H(x, w, \theta))\leq 1$ for all $(\theta, F)\in \Theta \times \mathcal{F}$, we have 
\[
p(1, x, w, \Gamma_{\epsilon}(\gamma))/p(1, x, w, \gamma^*)  \leq 1/\mathbb{P}_{F_0}(H(x, w, \theta_0)
\]
and
\[
p(0, x, w, \Gamma_{\epsilon}(\gamma))/p(0, x, w,  \gamma^*)  \leq 1/(1-\mathbb{P}_{F_0}(H(x, w, \theta_0))
\]
and consequently,
\begin{align*}
	\mathbb{E^*}\{[\log \{p(y,x, w, \Gamma_{\epsilon}(\gamma))/ & p(y, x, w, \gamma^*)\}]^+\}  \leq 
	\int  \Big\{-\mathbb{P}_{F_0}(H(z, w, \theta_0) \log \mathbb{P}_{F_0}(H(z, w, \theta_0)\\ 
	& - (1-\mathbb{P}_{F_0}(H(z, w, \theta_0)) \log (1-\mathbb{P}_{F_0}(H(z, w, \theta_0))\Big\} dG(z) < \infty
\end{align*}
where $G$ denotes the joint distribution of $(x,w)$. 
\end{proof}

\begin{proof}[Proof of Theorem \ref{thm: consistency}]
Let $\gamma \neq \gamma^*$. Note that $\log \{p(y, x, w, \Gamma_{\epsilon}(\gamma))/p(y, x, w, \gamma^*)\}$ 
is a monotone increasing function of $\epsilon$.  Lemma \ref{lemma: continuity} implies that 
$\underset{\epsilon \downarrow 0}{\lim} p(y, x, w, \Gamma_{\epsilon}(\gamma)) = p(y, x, w, \gamma)$, 
and dominated convergence implies that 
\[
\underset{\epsilon\downarrow0}{\lim} E^*\{[\log\{p(y,x, w,\Gamma_{\epsilon}(\gamma))/p(y,x, w, \gamma^*)\}]^+\} 
= E^*\{[\log \{ p(y,x, w,\gamma)/p(y,x, w,\gamma^*)\}]^+\}.
\]
Lemma \ref{lemma: continuity} and Fatou's Lemma then imply that 
\[
\underset{\epsilon\downarrow0}{\liminf} E^*\{[\log\{p(y,x, w,\Gamma_{\epsilon}(\gamma))/p(y,x, w, \gamma^*)\}]^-\}
\geq  E^*\{[\log \{ p(y,x, w,\gamma)/p(y,x, w,\gamma^*)\}]^-\}.
\]
Monotonicity of $\log \{p(y, x, w, \Gamma_{\epsilon}(\gamma))/p(y, x, w, \gamma^*)\}$ in $\epsilon$ 
ensures that the limit exists, so, 
\[
\underset{\epsilon\downarrow0}{\lim} E^*[\log\{p(y,x, w,\Gamma_{\epsilon}(\gamma))/p(y,x, w, \gamma^*)\}] 
\leq E^*[\log \{ p(y,x, w,\gamma)/p(y,x, w,\gamma^*)\}]<0.
\]
Strict inequality holds by Jensen's inequality and the identification result in 
Theorem \ref{thm: identification}.  Thus, for any $\gamma \neq \gamma^*$, there exists 
$\epsilon_{\gamma}>0$ such that 
	\[
	E^*[\log \{p(y,x, w,\Gamma_{\epsilon_{\gamma}}(\gamma))/p(y,x, w,\gamma^*)\}] < 0
	\]
	
Now for any $\epsilon>0$, the complementary set $(\Gamma_{\epsilon}(\gamma^*))^c$ is compact and 
is covered by $\cup_{\gamma \in (\Gamma_{\epsilon}(\gamma^*))^c} \Gamma_{\epsilon_{\gamma}}(\gamma)$, 
hence there exists a finite subcover, $\Gamma_1, \Gamma_2, \dots, \Gamma_J$ such that for each $j$,
\[
E^*[\log \{p(y,x, w,\Gamma_j)/p(y,x, w,\gamma^*)\}] <0.
\]
By the strong law of large numbers, as $n \to \infty$, 
\begin{align*}
\underset{\gamma \in (\Gamma_{\epsilon}(\gamma^*))^c}{\sup} \frac{1}{n}\sum_{i} 
\log \{p(y_i,x_i, w_i,\gamma) & /p(y_i,x_i, w_i, \gamma^*)\}\\
& \leq \underset{j=1,\dots,J}{\max} \frac{1}{n}\sum_i \log \{p(y_i,x_i, w_i,\Gamma_j)/p(y_i, x_i, w_i, \gamma^*)\}\\ 
& \overset{a.s}{\rightarrow} \max_{j=1,\dots,J} E^*[\log \{p(y, x, w, \Gamma_j)/p(y, x, w, \gamma^*)\}] < 0.
\end{align*}
Since $\epsilon$ is arbitrarily chosen, this implies that $(\hat \theta_n, \hat F_n) \to (\theta_0, F_0)$ almost surely when $n \to \infty$. 
Note that we allow $E^*[\log \{p(y, x, w, \Gamma_j)/p(y, x, w, \gamma^*)\}]$ to be minus infinity when invoking 
the strong law based on the following argument. For any $j$, 
\begin{align*}
	& \frac{1}{n}\sum_i \log \{p(y_i,x_i, w_i,\Gamma_j)/p(y_i, x_i, w_i, \gamma^*)\} \\
	& = \frac{1}{n}\sum_i \Big[\log \{p(y_i,x_i, w_i,\Gamma_j)/p(y_i, x_i, w_i, \gamma^*)\}\Big]^+ - \frac{1}{n}\sum_i \Big[\log \{p(y_i,x_i, w_i,\Gamma_j)/p(y_i, x_i, w_i, \gamma^*)\}\Big]^{-}.
\end{align*}
Lemma \ref{lemma: integrability} implies that 
\[
\frac{1}{n}\sum_i \Big[\log \{p(y_i,x_i, w_i,\Gamma_j)/p(y_i, x_i, w_i, \gamma^*)\}\Big]^+ \overset{a.s.}{\rightarrow} E^*\{[\log \{p(y_i,x_i, w_i,\Gamma_j)/p(y_i, x_i, w_i, \gamma^*)\}]^+\}.
	\]
It remains to show that 
\[
\frac{1}{n}\sum_i \Big[\log \{p(y_i,x_i, w_i,\Gamma_j)/p(y_i, x_i, w_i, \gamma^*)\}\Big]^{-} 
\overset{a.s.}{\rightarrow} E^*\{[\log \{p(y,x, w,\Gamma_j)/p(y,x, w,\gamma^*)\}]^{-}\}. 
\]
Suppose the right hand side is finite, then we can invoke the strong law.  
Alternatively, suppose $E^*\{[\log \{p(y,x, w,\Gamma_j)/p(y,x, w,\gamma^*)\}]^{-}\} = \infty$. 
Denote the random variable 
\[
\mathcal{R} := [\log \{p(y,x, w,\Gamma_j)/p(y,x, w,\gamma^*)\}]^{-},
\]
we have $\mathcal{R}\geq 0$ and $E \mathcal{R} = \infty$. Let $M$ be a constant, 
we have $E[\min\{\mathcal{R}, M\}] < \infty$ if $M < \infty$ and 
$E[\min\{\mathcal{Z}, M\}]  \to \infty$ as $M \to \infty$. Then for every $M$, 
	\[
	\frac{1}{n} \sum_i \mathcal{Z}_i \geq \frac{1}{n} 
	\sum_i \min \{\mathcal{Z}_i, M\} \overset{a.s.}{\rightarrow} E[\min \{\mathcal{Z}, M\}].
	\]
Assembling the foregoing, we have, 
	\[
	\frac{1}{n}\sum_i \Big[\log \{p(y_i,x_i, w_i,\Gamma_j)/p(y_i, x_i, w_i, \gamma^*)\}\Big]^{-} \overset{a.s.}{\rightarrow} E^*\{[\log \{p(y,x, w,\Gamma_j)/p(y,x, w,\gamma^*)\}]^{-}\}
	\]
	as required.
\end{proof}

\section{A Uniqueness Counterexample}\label{app.nonu}

\begin{figure}[bp]
    \begin{center}
    \resizebox{.65\textwidth}{!}{\includegraphics{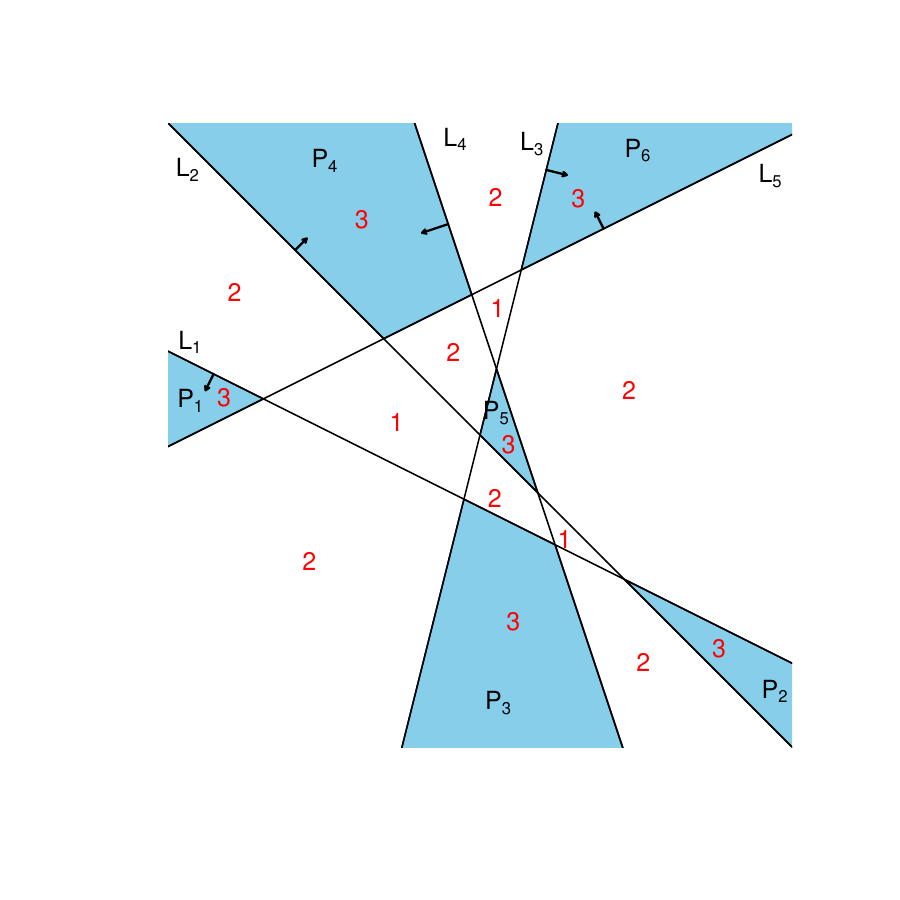}}
    \end{center}
    \caption{\small Non-uniqueness of the NPMLE:  Cell probabilities need not be uniquely
    determined by the NPMLE, despite the fact that half-space probabilties are unique.  The
    example illustrates 5 observations represented by the lines in the figure with orientation
    determined by the corresponding response values by the arrows.  Of the 16 polygons only 6
    are locally maximal and all 6 have the same cell count as indicated by the red numbers in
    the figure.  The NPMLE assigns equal weight to each of the 5 observations, but the mass
    assigned by the NPMLE to the 6 locally maximal polygons is not uniquely determined.
    }
    \label{fig.nonu}
\end{figure}

Uniqueness of the half-space probabilties produced by the NPMLE follows immediately from
the fact that a strictly concave objective is maximized over a convex set, however the
uniqueness of the cell probabilities is, as noted in the text, a more delicate matter.
In the context of interval censoring, \citeasnoun{GG94} suggest that non-uniqueness is
possible, and \citeasnoun{Lindsay.95} also suggests this possibility in the context of
general mixture models.  Neither, however, provides an example of this phenomenon,
so we felt it would be worthwhile to describe an explicit example where non-uniqueness occurs
in the context of our binary response model.  

Our example, illustrated in Figure \ref{fig.nonu}, involves 5 observations. 
In the notation of our prior exposition we have covariates
$z = (0.5, 1, -4, 3, -0.5)^\top$, and $v = (-0.25, 0, 0, 0.4, 0.5)^\top$ and response
$y = (0,1,1,0,0)^\top$. In the figure
the 5 lines are labeled $L_i: 1, \ldots , 5$, they yield a tessellation of the plane
into $1 + 5 + \binom{5}{2} = 16$ polygons since the lines are in general position.
The orientations of the half-spaces are determined by the $y$ observations and are represented
by the arrows in the upper portion of the figure.  Counts for each polygon are 
thus easily evaluated and appear in red;  there are 6 locally optimal cells shaded blue,
and labeled $P_j: j = 1, \ldots , 6$,
each with a count of 3.  The remaining neighboring cells need not be considered
further by the NPMLE.  The log-likelihood can be expressed as,
\begin{align*}
L(p) & = \sum_{i=1}^5 \log(q_i(p)) \\
&  = \log (p_1 + p_2 + p_3) + \log (p_2 + p_4 + p_5 + p_6) + \log (p_2 + p_3 + p_5 + p_6) \\
& + \log (p_1 + p_3 + p_4 + p_5) +  \log (p_1 + p_4 + p_6),
\end{align*}
and the NPMLE solves,
\[
\max_p \{ L(p) | Ap = q, p \in \SS_5 \}
\]
where $A$ is the Boolean matrix, 
\[
A = \begin{pmatrix}
    1 & 1 & 1 & 0 & 0 & 0\\
    0 & 1 & 0 & 1 & 1 & 1\\
    0 & 1 & 1 & 0 & 1 & 1\\
    1 & 0 & 1 & 1 & 1 & 0\\
    1 & 0 & 0 & 1 & 0 & 1
\end{pmatrix}
.
\]
It can be shown that the optimal solution puts equal mass on each of the 5 observations, i.e.,
$\hat q_i = 0.6$ for $i = 1, \ldots , 5$, however it can be also easily verified that any
$\hat p = (0.4 - k, 0.2, k, k, 0.2-k, 0.2)^\top$ for $k \in [0, 0.2]$ generates this solution.

It may also be noted that the Manski maximum score estimator for this example
consists of all 6 disconnected, shaded regions of Figure \ref{fig.nonu}, the most central of
which can be eliminated in favor of putting extra mass on polygon number 1 for the NPMLE.
Finally, it is important to observe that such examples with tied cell counts and $M > n$
are somewhat pathological.  In our experience the number of active cells, $M$, in the NPMLE
solution is generally less than $n$ and the cell probabilities are then uniquely determined from 
the half-space probabilities provided that the reduced $A$ matrix corresponding to the active
columns of the solution has full column rank.

\section{Computational Details}
All of the figures and tables presented here can be reproduced in the R language, \citeasnoun{R}, with code
provided by the second author.  Full algorithmic details and documentation are available in the R package
{\tt{RCBR}}, of \citeasnoun{RCBR}, which in turn relies upon the R packages {\tt{REBayes}} and {\tt{Rmosek}}
of \citeasnoun{REBayes} and \citeasnoun{Rmosek} respectively, and the Mosek Optimization Suite of
\citeasnoun{Mosek}.

\section{Supplementary Tables}
Two supplementary tables are provided in this section.  Table \ref{tab:GKTuning} reports log likelihood values for 
various choices of the tuning parameters of the Gautier-Kitamura estimator for the modal choice application.  The
contour plots for the Gautier-Kitamura estimates appearing in the main text are based on tuning parameters
maximizing log likelihood as reported in this table.  Table \ref{tab:car012}  reports the location and mass of the
NPMLE estimates for each subsample of the modal choice data;  only points with mass greater than 0.001 are
reported.  Note, once again, that locations are arbitrary interior points within the polygons optimizing
the log likelihood.
%latex.default(A, file = "GKTuning.tex", rgroup = paste(0:2, " Cars"),     n.rgroup = rep(7, 3), rowlabel = "T", caption = cap, caption.loc = "bottom",     label = "tab:GKTuning")%
\begin{table}[!tbp]
\begin{center}
\begin{tabular}{lrrrrr}
\hline\hline
\multicolumn{1}{l}{T}&\multicolumn{1}{c}{TX =  3}&\multicolumn{1}{c}{TX =  5}&\multicolumn{1}{c}{TX =  10}&\multicolumn{1}{c}{TX =  15}&\multicolumn{1}{c}{TX =  20}\tabularnewline
\hline
{\bfseries 0  Cars}&&&&&\tabularnewline
~~1&$ -39.33$&$ -41.27$&$ -40.94$&$ -41.54$&$ -41.13$\tabularnewline
~~2&$ -39.16$&$ -40.28$&$ -41.23$&$ -40.18$&$ -39.13$\tabularnewline
~~3&$ -39.54$&$ -40.64$&$ -40.75$&$ -39.51$&$ -39.07$\tabularnewline
~~4&$ -40.45$&$ -41.12$&$ -40.47$&$ -40.13$&$ -40.25$\tabularnewline
~~5&$ -41.29$&$ -41.83$&$ -41.02$&$ -41.19$&$ -41.58$\tabularnewline
~~7&$ -41.96$&$ -43.46$&$ -43.73$&$ -42.81$&$ -42.56$\tabularnewline
~~9&$ -42.81$&$ -45.48$&$ -47.13$&$ -45.96$&$ -45.25$\tabularnewline
\hline
{\bfseries 1  Cars}&&&&&\tabularnewline
~~1&$-223.33$&$-158.55$&$-178.90$&$-195.51$&$-163.35$\tabularnewline
~~2&$-179.44$&$-153.06$&$-157.35$&$-165.20$&$-153.68$\tabularnewline
~~3&$-145.51$&$-146.91$&$-162.93$&$-165.86$&$-166.01$\tabularnewline
~~4&$-148.49$&$-151.01$&$-162.31$&$-162.19$&$-180.44$\tabularnewline
~~5&$-151.37$&$-154.96$&$-169.23$&$-168.51$&$-198.44$\tabularnewline
~~7&$-161.38$&$-170.41$&$-191.81$&$-206.85$&$-229.93$\tabularnewline
~~9&$-168.29$&$-179.99$&$-198.82$&$-211.48$&$-230.02$\tabularnewline
\hline
{\bfseries 2  Cars}&&&&&\tabularnewline
~~1&$ -91.20$&$-164.31$&$-167.78$&$-126.56$&$ -93.70$\tabularnewline
~~2&$ -89.88$&$-126.42$&$-172.32$&$-143.76$&$-115.74$\tabularnewline
~~3&$-110.17$&$-141.13$&$-185.35$&$-165.72$&$-142.57$\tabularnewline
~~4&$-128.33$&$-152.74$&$-188.94$&$-162.99$&$-143.09$\tabularnewline
~~5&$-135.02$&$-154.10$&$-177.19$&$-148.85$&$-132.89$\tabularnewline
~~7&$-135.22$&$-146.45$&$-155.38$&$-132.70$&$-123.55$\tabularnewline
~~9&$-137.42$&$-145.49$&$-155.21$&$-139.99$&$-128.56$\tabularnewline
\hline
\end{tabular}
\caption{Log-likelihood of the Gautier-Kitamura estimator for various values of
the Fourier-Laplace series truncation parameters.}
\label{tab:GKTuning}
\end{center}
\end{table}

%latex.default(Tab, file = "car012.tex", rowlabel = "", cgroup = c("No Car",     "One Car", "Two Cars"), caption = cap, caption.loc = "bottom",     label = "tab:car012")%
\begin{table}[!tbp]
\begin{center}
\begin{tabular}{rrrcrrrcrrr}
\hline\hline
\multicolumn{3}{c}{\bfseries No Car}&\multicolumn{1}{c}{\bfseries }&\multicolumn{3}{c}{\bfseries One Car}&\multicolumn{1}{c}{\bfseries }&\multicolumn{3}{c}{\bfseries Two Cars}\tabularnewline
\cline{1-3} \cline{5-7} \cline{9-11}
\multicolumn{1}{c}{$\eta_1$}&\multicolumn{1}{c}{$\eta_2$}&\multicolumn{1}{c}{p}&\multicolumn{1}{c}{}&\multicolumn{1}{c}{$\eta_1$}&\multicolumn{1}{c}{$\eta_2$}&\multicolumn{1}{c}{p}&\multicolumn{1}{c}{}&\multicolumn{1}{c}{$\eta_1$}&\multicolumn{1}{c}{$\eta_2$}&\multicolumn{1}{c}{p}\tabularnewline
\hline
$-1.4300$&$ 0.0429$&$0.2743$&&$  0.6917$&$ 0.1217$&$0.1300$&&$ 1.1200$&$ 0.0000$&$0.5000$\tabularnewline
$-0.6625$&$-0.6700$&$0.1955$&&$  0.8446$&$-0.0008$&$0.1153$&&$ 1.1550$&$ 0.1400$&$0.2533$\tabularnewline
$-1.0942$&$ 0.0830$&$0.1194$&&$  9.7600$&$-0.2400$&$0.0999$&&$ 0.4495$&$ 0.0580$&$0.0918$\tabularnewline
$ 0.3900$&$ 0.0400$&$0.1099$&&$  0.6666$&$ 0.0081$&$0.0999$&&$ 0.4540$&$ 0.0143$&$0.0777$\tabularnewline
$-0.8019$&$ 0.0628$&$0.0757$&&$  0.6790$&$ 0.0430$&$0.0875$&&$-0.8889$&$ 0.0928$&$0.0254$\tabularnewline
$-1.0500$&$ 0.0520$&$0.0680$&&$ -0.1271$&$ 0.0385$&$0.0717$&&$-0.1170$&$ 0.0157$&$0.0216$\tabularnewline
$-8.8900$&$ 1.0750$&$0.0512$&&$  0.2500$&$ 0.0800$&$0.0624$&&$ 0.6216$&$-0.0045$&$0.0160$\tabularnewline
$-7.1750$&$ 0.9850$&$0.0482$&&$-12.7050$&$ 1.0350$&$0.0538$&&$-0.1280$&$ 0.0455$&$0.0123$\tabularnewline
$-0.1994$&$-0.0078$&$0.0437$&&$  0.9422$&$-0.0441$&$0.0475$&&$ 0.4450$&$ 0.0175$&$0.0019$\tabularnewline
$-0.4663$&$ 0.2275$&$0.0086$&&$ -0.0081$&$ 0.0588$&$0.0415$&&$$&$$&$$\tabularnewline
$-0.3177$&$ 0.0018$&$0.0055$&&$  0.5825$&$ 0.0650$&$0.0407$&&$$&$$&$$\tabularnewline
$$&$$&$$&&$ -1.0050$&$ 0.0967$&$0.0362$&&$$&$$&$$\tabularnewline
$$&$$&$$&&$  0.7086$&$ 0.0042$&$0.0346$&&$$&$$&$$\tabularnewline
$$&$$&$$&&$ 14.6300$&$-0.4400$&$0.0291$&&$$&$$&$$\tabularnewline
$$&$$&$$&&$ -0.2789$&$ 0.0536$&$0.0271$&&$$&$$&$$\tabularnewline
$$&$$&$$&&$  0.1650$&$ 0.0942$&$0.0196$&&$$&$$&$$\tabularnewline
$$&$$&$$&&$  0.8411$&$-0.0077$&$0.0027$&&$$&$$&$$\tabularnewline
$$&$$&$$&&$$&$$&$$&&$$&$$&$$\tabularnewline
$$&$$&$$&&$$&$$&$$&&$$&$$&$$\tabularnewline
$$&$$&$$&&$$&$$&$$&&$$&$$&$$\tabularnewline
\hline
\end{tabular}
\caption{Mass points of the estimated distribution of coefficients for commuters:  
    The first two columns in each panel indicate interior points of cells containing the estimated mass given by
    the third column of each panel.  Only mass points with mass greater than 0.001 are displayed.\label{tab:car012}}\end{center}
\end{table}

}}
\end{document}